\documentclass[11pt]{article}
\usepackage{fullpage}

\usepackage{amsmath,amsthm,amsfonts,amssymb,mathtools,bbm}
\usepackage{tikz, float, subcaption}
\usepackage{complexity}
\usetikzlibrary{arrows, calc, fit, positioning, shapes.geometric, quantikz2}
\usepackage[pdfpagelabels=false,colorlinks=true,allcolors=blue]{hyperref}
\usepackage{url}

\newif\ifobsolete
\obsoletefalse
\newenvironment{obsoleteenv}{\noindent\textbf{\textcolor{red}{Obsolete: }}}{
}
\newcommand{\obs}[1]{\ifobsolete{\begin{obsoleteenv}\color{red}#1\end{obsoleteenv}}\fi}

\newtheorem{thm}{Theorem}[section]
\newtheorem{conj}[thm]{Conjecture}

\newtheorem{claim}[thm]{Claim}

\newtheorem{lemma}[thm]{Lemma}

\newtheorem{cor}[thm]{Corollary}

\theoremstyle{definition}
\newtheorem{defn}[thm]{Definition}

\newtheorem{rmrk}[thm]{Remark}
\newtheorem{notation}[thm]{Notation}

\newcommand{\alphaitems}{\renewcommand{\labelenumi}{\alph{enumi}\/\mbox{)}}}

\usepackage{autonum}

\newcommand*{\nats}{\mathbb{N}}

\newcommand*{\complexes}{\mathbb{C}}
\newcommand*{\F}{\mathbb{F}}

\renewcommand*{\poly}{\mathit{poly}}

\newcommand*{\map}[3]{{{#1}:{#2}\rightarrow{#3}}}

\newcommand*{\st}{\;\middle\vert\;}
\newcommand*{\set}[1]{{\left\{#1\right\}}}
\newcommand{\sett}[2]{\left\{\,#1 \mid #2\,\right\}}
\newcommand*{\binary}{\set{0,1}}
\newcommand*{\bits}[1]{\binary^{#1}}

\newcommand*{\wt}{\textup{wt}}

\newcommand*{\ctrlX}[1]{{\textup{C}_{#1}X}}
\newcommand*{\ctrlZ}[1]{{\textup{C}_{#1}Z}}
\newcommand*{\cnot}{\textup{C-NOT}}
\renewcommand*{\p}{\varphi}
\newcommand*{\cH}{\mathcal{H}}
\newcommand*{\cJ}{\mathcal{J}}
\newcommand*{\calP}{\mathcal{P}}
\newcommand*{\cZ}{\mathcal{Z}}
\newcommand*{\mat}[1]{{\left[\begin{matrix}#1\end{matrix}\right]}}
\newcommand{\1}{\mathbf{1}}

\newcommand*{\nums}{\mathbb{N}}
\newcommand*{\field}{\mathbb{F}}
\newcommand*{\var}{\textup{var}}

\newcommand*{\two}{{\{0,1\}}}
\newcommand{\ones}{{\boldsymbol 1}}

\renewcommand{\a}{{\boldsymbol a}}
\renewcommand{\b}{{\boldsymbol b}}
\renewcommand{\c}{{\boldsymbol c}}
\newcommand{\w}{{\boldsymbol w}}
\newcommand{\x}{{\boldsymbol x}}
\newcommand{\y}{{\boldsymbol y}}
\newcommand{\z}{{\boldsymbol z}}

\title{Tight bounds on depth-$2$ $\QAC$-circuits computing parity}

\author{Daniel Pad\'e\hspace{0.625in}Stephen Fenner\\University of South Carolina\thanks{Computer Science and Engineering Department, Columbia, SC 29208 USA\@.   \texttt{djpade@gmail.com}, \texttt{fenner.sa@gmail.com}.  Part of the work was done while the first author visited the fourth author in June and July, 2019.} \and Daniel Grier\\IQC\thanks{Institute for Quantum Computing, University of Waterloo, Waterloo, ON N2L3G1 Canada. \texttt{daniel.grier@uwaterloo.ca}} \and Thomas Thierauf\\Ulm University\thanks{Department of Engineering, Computer Science and Psychology, Ulm, Germany.  \texttt{thomas.thierauf@uni-ulm.de}. Supported by DFG grant TH~472/5-2.}}

\begin{document}
\maketitle

\begin{abstract}
We show that the parity of more than three non-target input bits cannot be computed by $\QAC$-circuits of depth-$2$, not even uncleanly, regardless of the number of ancilla qubits.  This result is incomparable with other recent lower bounds on constant-depth $\QAC$-circuits by Rosenthal~[ICTS~2021,arXiv:2008.07470] and uses different techniques which may be of independent interest:
\begin{enumerate}
\item
We show that all members of a certain class of multivariate polynomials are irreducible.  
The proof applies a technique of Shpilka \& Volkovich [STOC 2008].
\item
We give a tight-in-some-sense characterization of when a multiqubit CZ gate creates or removes entanglement from the state it is applied to.
\end{enumerate}
The current paper strengthens an earlier version of the paper [arXiv:2005.12169].

\medskip

\end{abstract}

\section{Introduction}
\label{sec:intro}

Quantum decoherence is a major obstacle to maintaining long quantum computations.  Current quantum computers confront short decoherence times and so must act quickly to do useful computations, and this limitation is likely to continue long into the future.

A reasonable theoretical model of such computations are shallow quantum circuits, i.e., quantum circuits of small depth.  The decoherence dilemma has inspired much theoretical interest in the capabilities of these circuits, particularly circuits that have constant depth and polynomial size.  To solve useful problems, quantum circuits that are very shallow will require gates acting on several qubits at once. A major question then is this: do there exist multiple-qubit gates that are both potentially realizable and sufficient for powerful computation in small (even constant) depth?

It is known that, with the aid of \emph{fanout} gates (a certain multiqubit gate defined below), quantum circuits can do a variety of important tasks such as phase estimation and approximate Quantum Fourier Transform in essentially constant depth~\cite{HS:fanout}.  Are fanout gates necessary here?  If one only allows gates to act on $O(1)$ qubits each, it is clear that any decision problem computed by $o(\log n)$-depth quantum circuits with bounded error can only depend on $2^{o(\log n)}$ bits of the input (see Fang et al.~\cite{FFGHZ:fanout} for a discussion).  Thus without allowing \emph{some} class of quantum gates with unbounded width (arity), no nontrivial decision problem can be computed by such a circuit.
What if we restrict to constant-width quantum gates, but we allow measurement of several qubits at the end, followed by post-processing by a polynomial-size classical circuit?  Here the situation is more complicated.  For certain types of constant-depth circuits---particularly, for circuits with constant-width gates followed by a classical AND applied to the measured results of all the output qubits---one can compute in polynomial time the result, provided there is a wide enough gap in the probabilities of getting a $0$-result versus a $1$-result~\cite{FGHZ:constant-depth}.  In contrast, Bravyi, Gosset, \& K\"{o}nig presented a search problem\footnote{In a search problem (or relation problem) there may be several possible acceptable outputs, and the device is only required to produce one of them.} that can be computed exactly by a constant-depth quantum circuit with constant-width gates, and no classical probabilistic circuit of sublogarithmic depth can solve the same problem with high probability~\cite{BGK:quantum-advantage}.

Another type of multiqubit gate that has a natural definition is the quantum AND-gate, which flips the value of a target just when all the control qubits are on.\footnote{These gates are also called \emph{generalized Toffoli gates}.}  It is not clear whether such a gate will be easy to implement, but it is a natural question to compare the power of fanout versus quantum AND-gates with respect to constant-depth quantum computation.

A quantum circuit (actually a family of such circuits, one for each input size) using unbounded quantum AND-gates and single-qubit gates is called a \emph{$\QAC$-circuit}.  This is the quantum analogue of a classical $\AC$-circuit.    Takahashi \& Tani showed that the quantum AND-gate can be simulated exactly in constant depth by a quantum circuit with single-qubit gates and fanout gates~\cite{TT:constant-depth-collapse}.  The converse of the Takahashi \& Tani result---can a fanout gate be simulated exactly (or even approximately) by a constant-depth $\QAC$-circuit?---is still an open question, and is the main focus of this paper.  We conjecture that the answer is no, and our current results supply evidence in that direction, proving a separation between fanout and depth-2 $\QAC$-circuits.  It is known that quantum fanout gates are constant-depth equivalent to quantum parity gates~\cite{Moore:fanout}, and so the question at hand is a reasonable quantum analogue to the already proven separation between parity and $\AC^0$ in classical circuit complexity~\cite{Ajtai:AC0,FSS:AC0} (the superscript $0$ signifies constant-depth circuits).  This analogy is not perfect; in classical circuit complexity, fanout is usually taken for granted and used freely, and this is not the case with quantum circuits.

\begin{conj}\label{conj:QACneQACC}
Constant-depth $\QAC$-circuits cannot simulate quantum fanout gates.
\end{conj}

Partial progress on this conjecture
was made by Fang et al.~\cite{FFGHZ:fanout}, 
where it was shown that no constant-depth $\QAC$-circuit family (a.k.a.\ a $\QAC^0$-circuit family) \emph{with a sublinear number of ancilla qubits} can approximate a fanout gate.  Subsequent progress on this conjecture then stalled for several years.  
In 2014, E.~Pius~\cite{Pius:QAC} announced a result that parity (equivalently, fanout) of more than five qubits cannot be simulated cleanly by a $\QAC$-circuit with depth~2.\footnote{We ignore single-qubit gates in determining the depth of a circuit, counting only those layers containing multiqubit gates.  See Section~\ref{sec:prelims} for definitions.}  
We have been unable to verify his proof completely.  Nonetheless, some ideas in that paper have been helpful in a new push to prove the conjecture.

In an earlier version of our paper~\cite{PFGT:fanout} it was shown that no depth-$2$ $\QAC$-circuit on $n>3$ qubits can implement parity exactly.
This result improved upon that announced by Pius~\cite{Pius:QAC} by reducing the number of input qubits, and was tight in the sense that one can easily simulate the $3$-qubit parity gate cleanly with a depth-2 $\QAC$-circuit.

The current paper improves upon our earlier version~\cite{PFGT:fanout} by removing the cleanliness restriction, showing that no depth-2 $\QAC$-circuit can exactly compute the parity of more than three qubits, even uncleanly.  To do this, we introduce a new algebraic technique for our proof that is of independent interest and potentially useful for proving negative results for depth-$3$ and beyond.  Our technique is based on work of Shpilka \& Volkovich~\cite{SV:indecomposable} on variable-disjoint factors of a multivariate polynomial.  We show that a particular family of multivariate polynomials are all irreducible.  Using that, we prove a specific entangling property of the C-SIGN gate (a cousin of the generalized Toffoli gate; see Section~\ref{sec:prelims}).  Roughly, any essential application of a C-SIGN gate leaves all its qubits entangled, provided they were not so entangled to begin with.  By ``essential'' we mean that the gate does not disappear or simplify to a gate of smaller arity.

More recently and independently of us, improved bounds for depth-$d$ $\QAC$ circuits approximately computing $n$-qubit fanout/parity have been obtained by a number of people.  Rosenthal~\cite{Rosenthal:parity} proved that such circuits \emph{can} approximate parity arbitrarily closely when $d\ge 7$, albeit with an exponential number of ancilla qubits.  He also showed that depth-$2$ $\QAC$-circuits of arbitrary size cannot approximate parity (even uncleanly).  Nadimpalli, Parham, Vasconcelos, \& Yuen~\cite{NPVY:QAC0-Pauli-Spectrum} considered the Pauli spectra of polynomial-size $\QAC$ circuits, showing that such circuits of depth-$d$ using $n^{O(1/d)}$ ancilla qubits cannot compute parity on more than a $\left(\frac{1}{2} + 2^{-\Omega(n^{1/d})}\right)$-fraction of the inputs.  More recently, Anshu, Dong, Ou, \& Yao~\cite{ADOY:QAC0-superlinear-ancillae} obtained a slightly superlinear lower bound of $n^{1+3^{-d}}$ on the number of ancilla qubits needed to compute any function of linear approximate degree, including parity.  Improving this lower bound even slightly to $n^{1+\exp(-o(d))}$ would imply that parity is not in $\QAC^0$, leading to a separation of the language classes computed by these circuits: $\QAC^0 \ne \QACC^0$.  Here, $\QACC^0$-circuits are families of constant-depth, polynomial-size quantum circuits with single-qubit gates and unbounded mod-$q$ gates (for any $q>1$ constant across the circuits in the family).   (Parity gates were shown to be depth-$1$ equivalent to fanout gates~\cite{Moore:fanout}, so these circuits are layer-for-layer equivalent to circuits with fanout gates instead, and it is known that mod-$q$ gates are simulatable by $\QAC$-circuits with parity gates in constant depth, and vice versa~\cite{GHMP:QACC,HS:fanout,TT:constant-depth-collapse}.)  Echoing Rosenthal's result~\cite{Rosenthal:parity}, Grier \& Morris~\cite{GM:QAC-threshold} show that constant-depth, polynomial-size quantum circuits equipped with unbounded \emph{threshold} gates can compute fanout to arbitrarily close approximation.

Our results use techniques very different from all of those used above and address \emph{exact} computation of parity for non-asymptotic $n$, whereas those in~\cite{Rosenthal:parity,NPVY:QAC0-Pauli-Spectrum,ADOY:QAC0-superlinear-ancillae,GM:QAC-threshold} address approximations of various sorts that are asymptotic in nature.  Rosenthal's bounds on depth-$2$ circuits, for example, give asymptotic trade-offs between the closeness of the approximation and the maximum number of qubits allowed, and (as he implicitly admits) they are nontrivial only for $n$ at least roughly $10^{60,000}$.  Our current result is incomparable in that we give a \emph{tight} upper bound on $n$ allowing depth-$2$ circuits to \emph{exactly} compute parity (even uncleanly).

\section{Preliminaries}
\label{sec:prelims}

Let $n\in\nums$.
We define $[n]=\{1,\ldots,n\}$ and for
$s = s_1 s_2\cdots s_n \in \bits{n}$, 
we let $\wt(s)$ denote the \emph{Hamming weight} of~$s$, 
and  $\oplus s\in\binary$  the \emph{parity} of~$s$.
\obs{The \emph{bitwise complement} of~$s$ is denoted by~$\overline{s}$.}
\begin{align}
\wt(s) &= \sum_{k=1}^n s_i \\
\oplus s &= \wt(s) \bmod 2 \obs{\\
\overline{s} &= \overline{s_1}\, \overline{s_2}\, \cdots \, \overline{s_n}}
\end{align}

In a slight abuse of notation,
we use~$s$ also to denote the natural number in $[0,1,\dots,2^n-1]$ represented by $s$ in binary.  
(The correct interpretation will be clear from the context.)  
The binary string of length~$n$ of all $1$'s is denoted by~$\ones_n$.
If the length is clear from the context,
we sometimes write just~$\ones$.

\obs{
For a subset $I\subseteq [n]$ of the indices,
we define
$s|_I\in\two^n$ as the string that agrees with~$s$ on all bits in~$I$
and is set to~$1$ on all other bits
\[ 
(s|_I)_i = 
\begin{cases}
s_i & \text{if } i\in I, \\
1   & \text{if } i\notin I.
\end{cases} 
\]
So for example,
\[ 
(\overline{s|_I})_i = 
\begin{cases}
\overline{s}_i & \text{if } i\in I, \\
0   & \text{if } i\notin I.
\end{cases} 
\]
}

\begin{notation}\label{not:concat-convention}
Let $n = n_1 + n_2$.
For binary strings $s_1 \in \two^{n_1}$ and $s_2 \in \two^{n_2}$, 
we let $s = s_1\circ s_2 \in \two^{n}$ be the concatenation 
of~$s_1$ with~$s_2$.  
\end{notation}

\subsection{Algebraic Preliminaries}
\label{sec:alg-prelims}

For a field~$\field$ and variables $\x = (x_1,x_2, \dots, x_n)$,
let $\field[\x]$ denote the ring of $n$-variate polynomials over~$\field$.  
For $f\in\field[\x]$, 
an assigment $\a = (a_1,a_2\ldots,a_n) \in\field^n$, 
and a subset $I\subseteq [n]$, 
we define
$f|_{\x_I=\a}$ as the polynomial obtained from~$f$ by substituting~$a_i$ for~$x_i$, 
for each $i\in I$. 
Hence,
$f|_{\x_I=\a}$ is a polynomial in variables $x_j$, 
for $j\in \overline{I} = [n]\setminus I$.

We say that
$f\in\field[\x]$ \emph{depends on variable}~$x_i$, 
if there exists $\a\in\field^n$ such that 
$f|_{\x_{[n]\setminus\{i\}}= \a}$ is non-constant.
Note that $f|_{\x_{[n]\setminus\{i\}}=\a}$
is univariate in  variable~$x_i$.  
All variables that~$f$ depends on are denoted as~$\var(f)$,
\[
\var(f) = \sett{i\in [n]}{f \text{ depends on } x_i}.
\]
An assignment $\a\in \field^n$ that witnesses that~$f$ depends on all variables in~$\var(f)$ simultaneously
is called a justifying assignment for~$f$.

\begin{defn}[Justifying assignment]
For $f\in \field[\x]$, 
a \emph{justifying assignment} for~$f$ is an $\a\in \field^n$
such that
$f_{\x_{[n]\setminus\{i\}}= \a}$ depends on $x_i$, 
for \emph{all} $i\in\var(f)$.
\end{defn}

For fixed $f\in\field[\x]$, justifying assignments are known to exist provided $\field$ is big enough~\cite{SV:indecomposable},
in particular in infinite fields.

\begin{defn}
A polynomial $f\in\field[\x]$ is \emph{decomposable} 
if there exist nonconstant polynomials~$g,h$ over disjoint sets of variables such that $f = g h$.  
Otherwise $f$ is \emph{indecomposable}.
\end{defn}

Every polynomial $f\in\field[\x]$ factors uniquely (up to order and multiplication by nonzero elements of $\field$) into indecomposable factors over pairwise disjoint sets of variables. 
Since the factorization of a polynomial is unique
(up to the order of the factors),
the same holds for the decomposition. 
Irreducibility implies indecomposability.
The reverse implication holds for multilinear\footnote{By ``multilinear'' we mean that each variable has degree at most~$1$.} polynomials.  
Univariate polynomials are always indecomposable.  

A decomposition of $f\in\field[\x]$ induces a \emph{variable-partition}
of~$\var(f)$ by the factors,
where the sets correspond to the variables occurring in the indecomposable factors of~$f$.
Note that the partition is unique.
By convention, we extend this partition to a partition of $[n]$ by letting $\{i\}$ be an element of the partition for all $i\in [n]\setminus\var(f)$.

\subsection{Quantum Circuit Preliminaries}
\label{sec:q-circuit-prelims}


We write $z^*$ for the complex conjugate of $z\in\complexes$, and~$A^*$ for the adjoint (Hermitian conjugate) of an operator $A$ on a finite-dimensional Hilbert space.  Otherwise, our notation is fairly standard (see \cite{KLM:quantum-book,KSV:quantum-book,NC:quantumbook} for example).
For $m\ge 0$, we let $\cH_m$ denote the Hilbert space on $m$ qubits, labeled $1,\ldots,m$.  Thus $\cH_m$ has dimension $2^m$, and is isomorphic to $\left(\complexes^{2}\right)^{\otimes m}$ via the usual computational basis.
If $S$ is some subset of $[m]$, then we let $\cH_S$ denote the Hilbert space of the qubits with labels in $S$.  Thus for example, $\cH_m = \cH_{[m]}$.  For disjoint $S,T\subseteq [m]$, there is a natural isomorphism $\cH_{S\cup T} \cong \cH_S \otimes \cH_T$ obtained by merely permuting qubits as necessary, and so we will not distinguish between these two spaces.  For a subset $S$ of the qubits under consideration in the sequel, we let $\overline{S}$ denote the set of qubits not in $S$.

Our quantum circuit model with unitary gates is standard, found in several textbooks, including~\cite{NC:quantumbook,KLM:quantum-book}.  We assume our circuit acts on $\cH_m$ for some $m\in\nats$.  We assume qubits $1,\ldots,n$ are the \emph{input qubits}, for some $n\le m$, and the rest are \emph{ancilla qubits}.

All the quantum circuits we consider are allowed arbitrary single-qubit gates.  These gates do not count toward the depth of the circuit; only layers of multiqubit gates are counted for the depth.  For example, a depth-$1$ circuit many have multiqubit gates acting on disjoint sets of qubits simultanously (in a single layer), preceded and followed on each qubit with an arbitrary single-qubit gate.

The $1$-qubit Pauli gates are defined as usual:
\begin{align}
I &:= \mat{1&0\\0&1}\;, & X &:= \mat{0&1\\1&0}\;, & Y &:= \mat{0&-i\\i&0}\;, & Z &:= \mat{1&0\\0&-1}\;.
\end{align}
The $1$-qubit Hadamard gate is defined as $H := (X+Z)/\sqrt 2$.

The \emph{$k$-target fanout gate} $F_k$ acts on $k+1\ge 2$ qubits, where one qubit, the first, say, is the \emph{control} and the rest are targets:
\[ F_k\ket{c,x_1,x_2,\cdots,x_k} = \ket{c,\,c\oplus x_1,\,c\oplus x_2,\,\ldots,\,c\oplus x_k} \]
for all $c,x_1,\ldots,x_k\in\two$.  $F_k$ is equivalent to applying $k$ many $\cnot$-gates in succession, all with the same control qubit, and targets $1$ through $k$, respectively.  If the targets are initially all in the $\ket{0}$ state, then $F_k$ copies the classical value of the control qubit to each of the targets.\footnote{This does not violate the no-cloning theorem, because only the classical value is copied.}

The \emph{$k$-input parity gate} $\oplus_k$ acts on $k+1\ge 2$ qubits, where the first (say) is the target and the rest are control qubits:
\[ \oplus_k\ket{t,x_1,x_2,\ldots,x_k} = \ket{t\oplus x_1\cdots\oplus x_k,\,x_1,\,x_2,\,\ldots,\,x_k} \]
for any $t,x_1,\ldots,x_k\in\two$.  Note that we do not count the target as one of the inputs.  The parity gate $\oplus_k$ results from $F_k$ by conjugating each qubit with a Hadamard gate $H$, that is,
\[ \oplus_k = (H_1 H_2\cdots H_{k+1}) F_k (H_1 H_2\cdots H_{k+1}) \]
and vice versa \cite{Moore:fanout}.  We also use $\oplus_k$ to denote the classical Boolean parity function on $k$ input bits.

The \emph{$k$-qubit quantum AND-gate} (a.k.a.\ the generalized Toffoli gate) $\ctrlX{k}$ flips the value of the target (the first qubit, say) just when all control bits are $1$:
\[ \ctrlX{k}\ket{x_1,x_2,\ldots,x_k} = \ket{x_1\oplus {\textstyle\Pi_{j=2}^k x_j},\,x_2,\,\ldots,\,x_k} \]
for any $x_1,\ldots,x_k\in\two$.  For example $\ctrlX{2} = F_1 = \cnot$.

The gates mentioned above are all ``classical'' in the sense that they map computational basis states to computational basis states.  This is not true of the C-SIGN gate.

The \emph{$k$-qubit C-SIGN gate} $\ctrlZ{k}$ flips the overall phase just when all bits are $1$:
\[ \ctrlZ{k}\ket{x_1,\ldots,x_k} = (-1)^{\Pi_{j=1}^k x_j}\ket{x_1,\ldots,x_k} \]
for any $x_1,\ldots,x_k\in\two$.  The C-SIGN gate results from the quantum AND-gate by conjugating the target qubit with Hadamard gates: $\ctrlZ{k} = H_1(\ctrlX{k}) H_1$, and vice versa, $\ctrlX{k} = H_1(\ctrlZ{k}) H_1$.  A technical advantage of the C-SIGN gate over the quantum AND-gate is that the C-SIGN gate has no distinguished target or control qubits; all qubits incident to the gate are on the ``same footing;'' more precisely, the C-SIGN gate commutes with the SWAP operator applied to any pair of its qubits.  For example, we depict a $\ctrlZ{3}$-gate acting on adjacent qubits in a circuit diagram as follows:
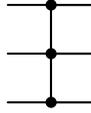
\begin{center}
\begin{quantikz}
& \control{} &\qw \\
& \ctrl{-1}  &\qw \\
& \ctrl{-1}  &\qw
\end{quantikz}
\end{center}
With that in mind we define, for any subset $S$ of the qubits of a multiqubit register, the gate $\ctrlZ{S}$ as the C-SIGN gate acting on the qubits in $S$.  Note, however, that $\ctrlZ{S}$ is a unitary operator on the entire register, being the tensor product of a C-SIGN gate on the qubits in $S$ with the identity operator on the other qubits.  We define $\ctrlZ{\emptyset} := -I$ by convention, where $I$ is the identity operator on the entire register.  We also refer to a C-SIGN gate acting on an unspecified set of qubits as a $\ctrlZ{}$-gate.

\begin{defn}
A \emph{$\QAC$-circuit} is a quantum circuit that includes $\ctrlZ{}$-gates and (arbitrary) single-qubit gates.  For $\QAC$-circuit~$C$, we define the \emph{depth} of~$C$ in the standard way, except we do not include single-qubit gates as contributing to the depth, i.e., as if all single-qubit gates are removed.
\end{defn}

\begin{defn}\label{def:gate-positions}
A depth-$d$ $\QAC$-circuit can have $d$ layers of $\ctrlZ{}$-gates, which we call \emph{layers~$1$ through $d$}, respectively, layer~$1$ lying to the left of layer~$2$, etc.  To the left, right, and in between these layers are arbitrary $1$-qubit gates.  Viewing the circuit as acting from left to right, the leftmost $1$-qubit gates are applied first; we say these gates are on layer~$0.5$.  Then the layer-$1$ $\ctrlZ{}$-gates are applied, followed by the $1$-qubit gates between layers~$1$ and $2$ (layer~$1.5$), followed by the $\ctrlZ{}$-gates on layer~$2$, and so on, then finally the rightmost layer of $1$-qubit gates (layer~$d+\frac{1}{2}$).

For a given layer~$\ell$ and qubit label~$j$, we denote by $G_j^{(\ell)}$ the gate on layer~$\ell$ that is incident to qubit~$j$.  If no such gate exists, then $G_j^{(\ell)} := I$.  Thus for integral $\ell$, \ $G_j^{(\ell)}$ is either $I$ or a $\ctrlZ{}$~gate, and for non-integral $\ell$, \ $G_j^{(\ell)}$ is a $1$-qubit gate.  For non-integral $\ell$, if $S$ is a subset of the qubits in the circuit, we let $G_S^{(\ell)}$ denote the tensor product $\bigotimes_{j\in S} G_j^{(\ell)}$ of the $G_j^{(\ell)}$ for $j\in S$.  For any $\ell$, we let $G^{(\ell)}$ denote the tensor product of all gates on layer~$\ell$, acting on all the qubits.
\end{defn}

Depending on the context, we can interpret $G_S^{(\ell)}$ as acting on the space $\cH_S$ or on the entire space, where it is then the tensor product of $G_S^{(\ell)}$ with the identity operator on the rest of the qubits.

\begin{defn}
If $G$ is an $n$-qubit unitary operator and $C$ is a quantum circuit on $n+m$ qubits for some $m\ge 0$, we say that \emph{$C$ cleanly simulates $G$} if, for all $x\in\two^n$,
\[ C(\ket{x}\otimes\ket{0^m}) = (G\ket{x})\otimes\ket{0^m}\;. \]
\end{defn}

So particularly, when the ancilla qubits are initially all $0$, they are returned to being all $0$ at the end.  We end this section by defining ways a quantum circuit computes a classical $1$-output function.

\begin{defn}\label{def:compute-f}
Let $\map{f}{\two^n}{\two}$ be a Boolean function, for some $n\ge 1$.  Let $\ket{\alpha}\in\cH_m$ be an $m$-qubit state, for some $m\ge 0$.  A quantum circuit $C$ on $1+n+m$ qubits \emph{$\ket{\alpha}$-computes} $f$ if, for all $x\in\two^n$, there exists an $(n+m)$-qubit state $\ket{\p_x}$ such that
\[ C(\ket{0}\otimes\ket{x}\otimes\ket{\alpha}) = \ket{f(x)}\otimes\ket{\p_x}\;. \]

We say that $C$ \emph{computes} $f$ if $C$ $\ket{0^m}$-computes $f$.

We say that $C$ \emph{weakly computes} $f$ if there exists $\ket{\alpha}\in\cH_m$ such that $C$ $\ket{\alpha}$-computes $f$.
\end{defn}

When using a quantum circuit to $\ket{\alpha}$-compute a function $f$ as in Definition~\ref{def:compute-f}, we label the qubits of $C$ with numbers from $0$ to $n+m$, with qubit~$0$ being the \emph{target}, qubits~$1,\ldots,n$ the \emph{input qubits}, and the qubits $n+1,\ldots,n+m$ the \emph{ancilla qubits}.  Note that for $\ket{\alpha}$-computing $f$, we do not make any ``cleanliness'' restrictions; we assume that the target starts in state $\ket{0}$ and ancilla qubits start in state $\ket{\alpha}$ and that the non-target qubits end in an arbitrary state.

Clearly, any circuit that cleanly simulates $\oplus_k$ also computes $\oplus_k$, and any circuit computing $\oplus_k$ also weakly computes $\oplus_k$.  Recall that we do not count the target qubit (qubit~$0$) as an input qubit, even though one could plausibly do this for the parity function.

\subsection{Representing Quantum States by Polynomials}
\label{sec:states-vs-polynomials}

Fix a $k\ge 1$ and let $\cH$ be a $k$-qubit Hilbert space.  $\cH$ has dimension $2^k$ with computational basis $\{ \ket{s} : s\in\two^k \}$ whose elements are indexed by binary strings of length $k$.  For each such basis state $\ket{s}$ we introduce a unique formal variable $x_s$ and define $\poly_\cH(\ket{s}) := x_s$.  The choice of the letter $x$ is not important and will depend on $\cH$.  We extend $\poly_\cH$ to all of $\cH$ by linearity, yielding a unique linear map $\map{\poly_\cH}{\cH}{\complexes[\{x_s :s\in\two^k\}]}$ so that, for any $v\in\cH$, writing $v = \sum_{s\in\two^k} \alpha_s\ket{s}$ for some coefficients $\alpha_s\in\complexes$, we have
\[ \poly_\cH(v) = \sum_{s\in\two^k} \alpha_s x_s\;. \]
(Here, $\complexes[S]$ is the ring of polynomials over variables in the set $S$.)  The map $\poly_\cH$ is clearly one-to-one, and its image is the set of homogeneous degree-$1$ polynomials in $\complexes[\{x_s\}]$.

Given a $k$-qubit Hilbert space $\cH$ and an $\ell$-qubit space $\cJ$, let $x_s := \poly_{\cH}(\ket{s})$ and $y_t := \poly_{\cJ}(\ket{t})$ for all $s\in\two^k$ and $t\in\two^\ell$.  The letters $x$ and $y$ are not important except that they must represent disjoint sets of variables.  We define $\map{\poly_{\cH,\cJ}}{\cH\otimes\cJ}{\complexes[\{x_s\}\cup\{y_t\}]}$ to be the unique linear map such that
\[ \poly_{\cH,\cJ}(\ket{s}\otimes\ket{t}) = \poly_{\cH}(\ket{s})\cdot\poly_{\cJ}(\ket{t}) = x_sy_t \]
for all $s\in\two^k$ and $t\in\two^\ell$.  Since the variable sets $\{x_s\}$ and $\{y_t\}$ are disjoint, we have that $\poly_{\cH,\cJ}$ is one-to-one.  It is also easily checked that for all $u\in\cH$ and $v\in\cJ$,
\begin{equation}\label{eqn:separability-implies-decomposability}
\poly_{\cH,\cJ}(u\otimes v) = \poly_\cH(u)\cdot \poly_\cJ(v)\;.
\end{equation}
Note that $\poly_{\cH,\cJ}$ is \emph{not} the same as $\poly_{\cH\otimes\cJ}$; the former maps to quadratic polynomials and the latter to linear polynomials.  We can extend this idea to tensor products of several spaces (we will need two, three, and four), choosing disjoint sets of variables for each: For example, letting $\cH_1,\cH_2,\cH_3,\cH_4$ be spaces of $k,\ell,m,n$ qubits, respectively, we define
\begin{equation}\label{eqn:poly-four}
\poly_{\cH_1,\cH_2,\cH_3,\cH_4}(\ket{s}\otimes\ket{t}\otimes\ket{u}\otimes\ket{v}) = x_sy_tz_uw_v
\end{equation}
for all binary strings $s,t,u,v$ of length $k,\ell,m,n$, respectively, where $\{x_s\},\{y_t\},\{z_u\},\{w_v\}$ are disjoint set of variables, and extend by linearity to all of $\cH_1\otimes\cH_2\otimes\cH_3,\otimes\cH_4$.


\section{Irreducibility Results}
\label{sec:irred-results}

In this section we present results used to prove the Entanglement Lemma (Lemma~\ref{lem:S-entangled}), which in turn is used to prove our depth-$2$ $\QAC$-circuit lower bound (Theorem~\ref{thm:depth-2-unclean}).  The results here may be of independent interest, however, and can potentially be used to prove stronger versions of the lemma and stronger circuit lower bounds.
We state the lemmas for field $\F = \complexes$,
but they hold as well over any sufficiently large field.\footnote{Some trivial modifications are needed for fields with characteristic~$2$ or $3$.}

Shpilka and Volkovich~\cite{SV:indecomposable}
gave a characterization of when a set $I\subseteq [n]$
is a union of sets from the variable-partition of~$f$.

\begin{lemma}[{\cite[Lemma 3.2]{SV:indecomposable}}]\label{lem:sv}
Let $f \in \field[\x]$ be a polynomial and let $\a\in\field^n$ be a justifying assignment of~$f$.  
Then $I\subseteq [n]$ satisfies  
$f(\a) \cdot f \equiv f|_{\x_I=\a}\cdot f|_{\x_{[n]\setminus I}=\a}$, 
if and only if 
$I$ is a union of sets from the variable-partition of~$f$.
\end{lemma}

We will use the following consequence of Lemma~\ref{lem:sv}.

\begin{cor}\label{cor:zero-assignment}
Let $f\in\field[\x]$ be a polynomial.  
If there exists a justifying assignment~$\a$ of~$f$ such that $f(\a) = 0$, 
then~$f$ is indecomposable.
\end{cor}

\begin{proof}
For simplicity of notation let $\var(f) = [n]$.  
Let $\a$ be any justifying assignment of~$f$.  
Suppose~$f$ decomposes into $f = gh$,
where $g,h$ are non-constant and variable-disjoint.
Let $I = \var(g)$.
Then~$I \not= \emptyset$ is the disjoint union of sets from the variable-partition 
of~$f$ and 
$\var(h) = [n] \setminus I$.
Hence,
by Lemma~\ref{lem:sv} we have 
\[
f(\a)\cdot f \equiv f|_{\x_I=\a}\cdot f|_{\x_{[n]\setminus I}= \a}.
\]  
Because $\a$ is justifying,
we have $f|_{\x_I=\a}\not\equiv 0$ and $f|_{\x_{[n]\setminus I}= \a} \not\equiv 0$.
Therefore $f(\a)\cdot f \not\equiv 0$, whence $f(\a)\ne 0$.
\end{proof}

We apply Corollary~\ref{cor:zero-assignment} to prove Lemma~\ref{lem:all-contact-two-zeros}, below.  That and the next two lemmas (Lemmas~\ref{lem:all-contact-one-zero} and \ref{lem:all-contact}) will only be used to prove analogous but more general lemmas (Lemmas~\ref{lem:most-general-two-zeros}, \ref{lem:most-general-one-zero}, and \ref{lem:most-general}) that we will use in Section~\ref{sec:entanglement-lemma}.

\begin{lemma}\label{lem:all-contact-two-zeros}
Let $k,m\in\nums$ be positive.  Define the polynomials
\begin{align}
T_1(\x) &= \sum_s c_s x_s & T_2(\z) &= \sum_u d_u z_u\;,
\end{align}
where
\begin{itemize}
\item
$s$ and $u$ run over the strings in $\two^k$ and $\two^m$, respectively,
\item
$x_s$ is a variable for each $s \in \two^k$ and $z_u$ is a variable for each $u \in \two^m$, and
\item
$c_s, d_u\in\complexes$ are coefficients such that
\begin{itemize}
\item
$c_{\ones} \not= 0$ ~and~ $d_{\ones} \not= 0$, 
\item
$\exists s \ne\ones~~ c_s\ne 0$,
\item
$\exists u \ne\ones~~ d_u\ne 0$.
\end{itemize}
\end{itemize}
Fix a nonzero $\alpha\in\complexes$ and define
\begin{equation}\label{eqn:C-generic-two-zeros}
P = T_1T_2 - \alpha c_{\ones}d_{\ones}x_\ones z_\ones\,.
\end{equation}
Then $P$ is indecomposable and hence irreducible.
\end{lemma}

\begin{proof}
We find a justifying assignment $\a = \a(A,B)$ of $P$ such that $P(\a) = 0$, 
satisfying Corollary~\ref{cor:zero-assignment}, where $\a$ depends on values $A\in\{1,2,3,4,5\}$ and $B\in\complexes$ is yet to be determined.
Fix $s_0\ne\ones$ and $u_0\ne\ones$ such that $c_{s_0} \ne 0$ and $d_{u_0} \ne 0$.  
We define~$\a$ by
the following assignment to the $\x$- and $\z$-variables:
\begin{align}
x_s &:= \begin{cases}
A & \mbox{if $s = s_0$,}\\
1 & \mbox{if $s = \ones$,} \\
0 & \mbox{otherwise,}
\end{cases}
&
z_u &:= \begin{cases}
B & \mbox{if $u = u_0$,}\\
1 & \mbox{if $u = \ones$,} \\
0 & \mbox{otherwise.}
\end{cases}
\end{align}
This makes
\begin{align}
T_1 &= c_{\ones} + c_{s_0}A\;, & T_2 &= d_{\ones} + d_{u_0}B\;, \label{eqn:Ts-AB}
\end{align}
and
\begin{equation}\label{eqn:root}
P(\a) = T_1T_2 - \alpha c_{\ones}d_{\ones}\;.
\end{equation}
We consider the projections of~$P$ to univariate polynomials,
for every variable of~$P$,
where the other variables are set according to~$\a$.
For the $\x$- and $\z$-variables,
let the projections be polynomials~$P_s^{(1)}(x_s)$
and $P_u^{(2)}(z_u)$.
We have
\begin{align}
P_s^{(1)}(x_s) &= 
\begin{cases}
(T_2 - \alpha d_{\ones})\, c_{\ones} x_\ones  + C_\ones, & \text{for } s = \ones, \\[2ex]
 c_s T_2x_s  + C_s, & \text{for } s\ne\ones , \label{eqn:xs-two-zeros} 
\end{cases}
\\
P_u^{(2)}(z_u) &= 
\begin{cases}
(T_1 - \alpha c_{\ones})\, d_{\ones} z_\ones + D_\ones, & \text{for } u = \ones,  \\[2ex]
T_1 d_u z_u + D_u, & \text{for } u\ne\ones , \label{eqn:zu-two-zeros}
\end{cases}
\end{align}
for constants $C_s,D_u \in \complexes$.

We choose $A,B$ such that 
for the assignment $\a = \a(A,B)$,
all the polynomials~$P_s^{(1)}(x_s)$
and~$P_u^{(2)}(z_u)$ are nonconstant and  $P(\a) = 0$.  
By Eqs.~(\ref{eqn:xs-two-zeros},\ref{eqn:zu-two-zeros}),
we must have
\begin{align}
T_1,T_2 &\ne 0 , & T_2 &\ne \alpha d_{\ones},  & T_1 &\ne \alpha c_{\ones}.
\end{align}
By Eq.~(\ref{eqn:Ts-AB}),
this excludes two values for each of~$A$ and~$B$.

Setting $P(\a) = 0$ and using Eqs.~(\ref{eqn:Ts-AB},\ref{eqn:root}),
we get
\begin{equation}\label{eqn:P(a)=0}
(c_{\ones} + c_{s_0}A)\, (d_{\ones} + d_{u_0}B )  = \alpha c_{\ones}d_{\ones}.
\end{equation}
Now observe that for any~$A$ such that $T_1 = c_{\ones} + c_{s_0}A \not= 0$,
there is a unique~$B$ that fulfills~(\ref{eqn:P(a)=0}),
namely
\begin{equation}\label{eqn:B-from-A}
B = \frac{\alpha c_{\ones}d_{\ones}}{d_{u_0}(c_{\ones} + c_{s_0}A)} - \frac{d_{\ones}}{d_{u_0}}.
\end{equation}
Moreover the mapping of~$A$ to solution~$B$ is injective.
Recall that we have to avoid two values for each of~$A$ and~$B$.
Hence,
when we select~$A$ out of~$5$ values, say $A \in \{0,1,2,3,4\}$,
one of the five values for~$A$ must give an appropriate~$B$
according to~(\ref{eqn:B-from-A})
such that
$\a = \a(A,B)$ is a justifying assignment for~$P$ and $P(\a) = 0$.
\end{proof}


The next two lemmas extend the polynomial~$P$ in Lemma~\ref{lem:all-contact-two-zeros} to more variables, but still being multilinear.
The first extension introduces $\w$-variables in~$T_2$.

\begin{lemma}\label{lem:all-contact-one-zero}
Let $k,m,n\in\nums$ be positive.  Define the polynomials
\begin{align}
T_1(\x) &= \sum_s c_s x_s & T_2(\z,\w) &= \sum_{u,v} d_{u,v} z_uw_v\;,
\end{align}
where
\begin{itemize}
\item
$s$, $u$, and $v$ run over the strings in $\two^k$, $\two^m$, and $\two^n$, respectively,
\item
$x_s$ is a variable for each $s\in\two^k$ and
similarly for the $z_u$ and $w_v$, and
\item
$c_s, d_{u,v}\in\complexes$ are coefficients such that
\begin{itemize}
\item
$c_{\ones} \not= 0$ ~and~ $d_{\ones,\ones} \not= 0$, 
\item
$\exists s\ne\ones~~  c_s\ne 0$,
\item
$\exists u \ne\ones ~ \exists v ~~ d_{u,v}\ne 0$ ~and~ 
$\exists u ~ \exists v\ne\ones ~~ d_{u,v}\ne 0$.
\end{itemize}
\end{itemize}
Fix a nonzero $\alpha\in\complexes$ and define
\begin{equation}\label{eqn:C-generic-one-zero}
P = T_1T_2 - \alpha c_{\ones}d_{\ones,\ones}x_\ones z_\ones w_\ones\;.
\end{equation}
Then $P$ is indecomposable and hence irreducible.
\end{lemma}

\begin{proof}
We define an assignment for the $\w$-variables
such that~$P$ gets projected to the $\x$- and $\z$-variables
and fulfills the assumptions from Lemma~\ref{lem:all-contact-two-zeros}.
Then we can conclude that~$P$ is indecomposable.

Let $u_0 \not= \ones$ and $v_0$ be such that $d_{u_0,v_0} \ne 0$.
For $r = 0,1,2$ define
\begin{equation}
b_v(r) = 
\begin{cases}
1, & \text{for } v = v_0,\\
r, & \text{for } v = \ones,\\
0, & \text{otherwise}.
\end{cases}
\end{equation}
Define $d_u(r) = d_{u,v_0} + r d_{u,\ones}$.
Then we have
\begin{equation}
T_2(\z,\b(r)) = \sum_{u} d_u(r)\, z_u. 
\end{equation}
We next show that there is an $r \in \{0,1,2\}$ such that
$T_2(\z,\b)$ fulfills the assumption in Lemma~\ref{lem:all-contact-two-zeros},
i.e.,
$d_{u_0}(r) \ne 0$
and
$d_{\ones}(r) \ne 0$:
\begin{itemize}
\item
For $r=0$, we have $d_{u_0}(0) = d_{u_0,v_0} \ne 0$ by assumption.
If $d_{\ones}(0) = d_{\ones,v_0} \ne 0$,
then $r=0$ works.

\item
So suppose now that $d_{\ones,v_0} = 0$.
Then we consider $r=1$.
We have $d_{\ones}(1) = d_{\ones,\ones} \ne 0$ by assumption.
If $d_{u_0}(1) = d_{u_0,v_0} + d_{u_0,\ones} \ne 0$,
then we may choose $r=1$.
 \item
So suppose now that  $d_{\ones,v_0} = 0$ and $d_{u_0}(1) = 0$.
Then we consider $r=2$.
We still have $d_{\ones}(2) = 2 d_{\ones,\ones} \ne 0$.
And now
also $d_{u_0}(2) = d_{u_0,v_0} + 2d_{u_0,\ones} = d_{u_0,\ones} = -d_{u_0,v_0} \ne 0$. 
\end{itemize}

For this~$r$, 
define $d_u = d_u(r)$ and $\b = \b(r)$ and
\begin{equation}\label{eq:P'xz}
P'(\x,\z) = P(\x,\z,\b)\;.
\end{equation}
Then $P'$ is an indecomposable polynomial by
Lemma~\ref{lem:all-contact-two-zeros}.

Assume that~$P$ is decomposable.
That is, we can write $P = gh$ for non-constant polynomials~$g,h$
on disjoint set of variables.
By~(\ref{eq:P'xz}),
we conclude that
\[
P'(\x,\z) = g|_{\w=\b}\, h|_{\w=\b}.
\]
Since
$P'$ is indecomposable,
it follows that one of the two factors is constant, say~$g|_{\w=\b}$.
Hence,
$g$ depends only on $\w$-variables, $g \in \F[\w]$.
Thus we can write
\begin{equation}\label{eq:P}
P(\x,\z,\w) = g(\w)\, h(\x,\z,\w)\;.
\end{equation}

Define similarly as above for~$\w$ an assignment~$\b'$ for~$\z$
such that
$T_2(\b',\w)$ fulfills the assumption in Lemma~\ref{lem:all-contact-two-zeros}.
Then
\begin{equation}\label{eq:P''}
P''(\x,\w) = P(\x,\b',\w) 
\end{equation}
is an indecomposable polynomial by
Lemma~\ref{lem:all-contact-two-zeros}.
But by~(\ref{eq:P}) we have
\begin{equation}
P''(\x,\w) = g(\w)\, h(\x,\b',\w),
\end{equation}
a contradiction.
\end{proof}


The second extension of 
Lemma~\ref{lem:all-contact-one-zero}
can also be seen as an extension of 
Lemma~\ref{lem:all-contact-one-zero}
where we introduce $\y$-variables for~$T_1$.

\begin{lemma}\label{lem:all-contact}
Let $k,\ell,m,n\in\nums$ be positive.  Define the multilinear polynomials
\begin{align}
T_1(\x,\y) &= \sum_{s,t} c_{s,t} x_sy_t & 
T_2 (\z,\w) &= \sum_{u,v} d_{u,v} z_u w_v\;,
\end{align}
where
\begin{itemize}
\item
$s$, $t$, $u$, and $v$ run over the strings in $\two^k$, $\two^\ell$, $\two^m$, and $\two^n$, respectively,
\item
$x_s$ is a variable for each $s\in\two^k$ and similarly for the $y_t$, $z_u$, and $w_v$, and
\item
$c_{s,t}, d_{u,v}\in\complexes$ are coefficients such that
\begin{itemize}
\item
$c_{\ones,\ones} \not=0$ ~and~ $d_{\ones,\ones} \not= 0$, 
\item
$\exists s \not= \ones ~\exists t~~ c_{s,t}\ne 0$ ~and~ 
$\exists s ~\exists t\not= \ones ~~c_{s,t}\ne 0$,
\item
$\exists u \not= \ones ~ \exists v~~ d_{u,v}\ne 0$ ~and~ 
$\exists u ~\exists v \ne\ones ~~  d_{u,v}\ne 0$.
\end{itemize}
\end{itemize}
For any $0 \not= \alpha\in\complexes$, define polynomial~$P(\x,\y,\z,\w)$ as
\begin{equation}\label{eqn:C-generic}
P = T_1T_2 - \alpha\, c_{\ones,\ones}d_{\ones,\ones}x_\ones y_\ones z_\ones w_\ones\;.
\end{equation}
Then $P$ is indecomposable and hence irreducible.
\end{lemma}

The proof is completely analogous to the proof of 
Lemma~\ref{lem:all-contact-one-zero}.
There we have seen a reduction to Lemma~\ref{lem:all-contact-two-zeros}.
Here we can similarly reduce to the case of 
Lemma~\ref{lem:all-contact-one-zero}.


The next lemma generalizes Lemma~\ref{lem:all-contact-two-zeros}.  Lemma~\ref{lem:all-contact-two-zeros} is the special case of Lemma~\ref{lem:most-general-two-zeros} where $k_2=m_2=0$.

\begin{lemma}\label{lem:most-general-two-zeros}
Let $k = k_1 + k_2$ and $m = m_1 + m_2$,
where $k_1,m_1 \geq 1$.
Define the polynomials
\begin{align}
T_1(\x) &= \sum_s c_s x_s & T_2(\z) &= \sum_u d_u z_u\;,
\end{align}
where
\begin{itemize}
\item
$s$ and $u$ run over the strings in $\two^k$ and $\two^m$, respectively,
\item
$x_s$ is a variable for each $s \in \two^k$ and $z_u$ is a variable for each $u \in \two^m$, and
\item
$c_s, d_u\in\complexes$ 
such that for $s = s_1 \circ s_2$,
where $s_1 \in \two^{k_1}$ and $s_2 \in \two^{k_2}$,
and
$u = u_1 \circ u_2$,
where $u_1 \in \two^{m_1}$ and $u_2 \in \two^{m_2}$,
\begin{itemize}
\item
$\exists s_2 ~~ c_{\ones \circ s_2} \not= 0$
~and~
$\exists u_2 ~~ d_{\ones \circ u_2} \not= 0$
\item
$\exists s_1 \ne\ones ~ \exists s_2 ~~ c_s \ne 0$,
\item
$\exists u_1 \ne\ones ~ \exists u_2 ~~ d_u\ne 0$.
\end{itemize}
\end{itemize}

Fix a nonzero $\alpha\in\complexes$ and define
\[ P = T_1T_2 - \alpha\sum_{s:s_1=\ones}\,\sum_{u:u_1=\ones} c_s d_u x_s z_u\;. \]
Then $P$ is indecomposable and hence irreducible.
\end{lemma}

\begin{proof}
We define a partial assignment to the $\x$- and $\z$-variables
so that~$P$ gets projected to the form in 
Lemma~\ref{lem:all-contact-two-zeros}.

Recall that each index~$s$ of an $\x$-variable
is split as $s = s_1 \circ s_2$,
where $s \in \{0,1\}^k$,
$s_1 \in \{0,1\}^{k_1}$, and  $s_2 \in \{0,1\}^{k_2}$,
where $k = k_1 + k_2$.
By our assumptions,
there are 
$\dot{s}_2, \ddot{s}_2 \in \{0,1\}^{k_2}$
such that
$c_{\ones \circ \dot{s}_2} \not= 0$ and 
$c_{s_1 \circ \ddot{s}_2} \not= 0$, 
for some $s_1 \in \{0,1\}^{k_1}$
such that $s_1\ne\ones$.
Now the projection is defined as follows:
We maintain one variable for each $s_1 \in \{0,1\}^{k_1}$,
namely $x_{\ones \circ \dot{s}_2}$,
and 
$x_{s_1 \circ \ddot{s}_2}$, 
for $s_1 \not= \ones$.
All other $\x$-variables we set to~$0$.
Let~$\b$ be this assignment.
Similarly,
we project the $\z$-variables via an assignment~$\c$
in an analogous way.

Observe that $P'= P|_{\x=\b,\z = \c}$ is of the form
as in Lemma~\ref{lem:all-contact-two-zeros} and
fulfills the assumptions made there.
Hence, 
we have that $P'$ is indecomposable.

Now assume that $P$ is decomposable, 
That is, we can write $P = gh$, for non-constant polynomials~$g, h$
on disjoint sets of variables. 
Hence, for $g'= g|_{\x=\b,\z = \c}$ and $h'= h|_{\x=\b,\z = \c}$,
we have that
\begin{equation}\label{eq:P'general}
P'= g' h'.
\end{equation}
Since $P'$ is indecomposable,
it follows that one of the two factors, say~$g'$, is constant.
Hence,
$g$ depends only on the variables that are set to~$0$ by 
assignments~$\b$ and~$\c$.
However,
since~$P$ is a homogeneous polynomial,
the factors~$g$ and~$h$ are homogeneous as well,
and therefore $g'= 0$.
But this contradicts~(\ref{eq:P'general}).
\end{proof}


The next two lemmas generalize
Lemma~\ref{lem:all-contact-one-zero} and~\ref{lem:all-contact}
in the same way as Lemma~\ref{lem:most-general-two-zeros}
generalizes Lemma~\ref{lem:all-contact-two-zeros}. 
Their proofs follow the proof of 
Lemma~\ref{lem:most-general-two-zeros} almost literally,
so we omit them.

\begin{lemma}\label{lem:most-general-one-zero}
Let $k = k_1 + k_2$, $m = m_1 + m_2$, and $n = n_1 + n_2$,
where $k_1,m_1, n_1 \geq 1$.
Define the polynomials
\begin{align}
T_1(\x) &= \sum_s c_s x_s & T_2(\z,\w) &= \sum_{u,v} d_{u,v} z_uw_v\;,
\end{align}
where
\begin{itemize}
\item
$s$, $u$, and $v$ run over the strings in $\two^k$, $\two^m$, and $\two^n$, respectively,
\item
$x_s$ is a variable for each $s\in\two^k$ and
similarly for the $z_u$ and $w_v$, and

\item
$c_s, d_{u,v}\in\complexes$
such that for $s = s_1 \circ s_2$,
where $s_i \in \two^{k_i}$,
and
$u = u_1 \circ u_2$,
where $u_i \in \two^{m_i}$,
and
$v = v_1 \circ v_2$,
where $v_i \in \two^{n_i}$, for $i =1,2$, we have
\begin{itemize}
\item
$\exists s_2 ~~ c_{\ones \circ s_2} \not= 0$
~and~ 
$\exists u_2,v_2 ~~ d_{\ones \circ u_2,\ones \circ v_2} \not= 0$
\item
$\exists s_1 \ne\ones ~ \exists s_2 ~~ c_s \ne 0$,
\item
$\exists u_1 \ne\ones ~ \exists u_2 ~ \exists v ~~ d_{u,v}\ne 0$
~and~ $\exists u ~ \exists v_1 \ne\ones ~ \exists v_2 ~~ d_{u,v}\ne 0$.
\end{itemize}

\end{itemize}
Fix a nonzero $\alpha\in\complexes$ and define
\[ P = T_1T_2 - \alpha\sum_{s:s_1=\ones}\,\sum_{u:u_1=\ones}\,\sum_{v:v_1=\ones} c_sd_{u,v}x_sz_uw_v\;. \]
Then $P$ is indecomposable and hence irreducible.
\end{lemma}


Finally,
we generalize Lemma~\ref{lem:all-contact} similarly.

\begin{lemma}\label{lem:most-general}
Let $k = k_1 + k_2$,  $\ell = \ell_1 + \ell_2$, $m = m_1 + m_2$,
and $n = n_1 + n_2$
where $k_1,m_1, l_1,n_1 \geq 1$.
Define the polynomials
\begin{align}
T_1(\x,\y) &= \sum_{s,t} c_{s,t} x_sy_t & T_2(\z,\w) &= \sum_{u,v} d_{u,v} z_uw_v\;,
\end{align}
where
\begin{itemize}
\item
$s$, $t$, $u$, and $v$ run over the strings in $\two^k$, $\two^\ell$, $\two^m$, and $\two^n$, respectively,
\item
$x_s$ is a variable for each $s\in\two^k$ and similarly for the $y_t$, $z_u$, and $w_v$, and
\item
$c_{s,t}, d_{u,v}\in\complexes$ such that
such that for $s = s_1 \circ s_2$,
where $s_i \in \two^{k_i}$,
and
$t = t_1 \circ t_2$,
where $t_i \in \two^{\ell_i}$,
and
$u = u_1 \circ u_2$,
where $u_i \in \two^{m_i}$,
and
$v = v_1 \circ v_2$,
where $v_i \in \two^{n_i}$, for $i =1,2$, we have
\begin{itemize}
\item
$\exists s_2,t_2 ~~ c_{\ones \circ s_2,\ones \circ t_2} \not= 0$
~and~
$\exists u_2,v_2 ~~ d_{\ones \circ u_2,\ones \circ v_2} \not= 0$
\item
$\exists s_1 \ne\ones ~ \exists s_2 ~ \exists t ~~ c_{s,t} \ne 0$ ~and~
$\exists s ~ \exists t_1 \ne\ones ~ \exists t_2 ~~ c_{s,t} \ne 0$, 
\item
$\exists u_1 \ne\ones ~ \exists u_2 ~ \exists v ~~ d_{u,v}\ne 0$ ~and~ 
$\exists u ~ \exists v_1 \ne\ones ~ \exists v_2 ~~ d_{u,v}\ne 0$.
\end{itemize}

\end{itemize}
Fix a nonzero $\alpha\in\complexes$ and define
\begin{equation}\label{eqn:C-no-zeros}
P = T_1T_2 - \alpha\sum_{s:s_1=\ones}\,\sum_{t:t_1=\ones}\,\sum_{u:u_1=\ones}\,\sum_{v:v_1=\ones} c_{s,t}d_{u,v}x_sy_tz_uw_v\;.
\end{equation}
Then $P$ is indecomposable and hence irreducible.
\end{lemma}


\section{The Entanglement Lemma}
\label{sec:entanglement-lemma}


Sets $A,B \subseteq [r]$ are a \emph{bipartition of}~$[r]$,
if $A,B \not= \emptyset$,
 $A \cup B = [r]$, and $A \cap B = \emptyset$.

\begin{defn}[Separable and $S$-separable states]\label{def:separates-at}
Suppose we have an $r$-qubit register with qubits labeled $1,\ldots,r$.  Let $\ket{\psi}$ be some state of the $r$~qubits, and let $A,B\subseteq [r]$ be a bipartition of~$[r]$.
State~$\ket{\psi}$ \emph{separates at $\{A,B\}$}, if $\ket{\psi} = \ket{\psi}_A\otimes\ket{\psi}_B$, for some $\ket{\psi}_A\in\cH_A$ and $\ket{\psi}_B\in\cH_B$.

Let $S\subseteq [r]$ be a subset of the qubits with $|S| \ge 2$.  
State~$\ket{\psi}$ is \emph{$S$-separable}, if $\ket{\psi}$ separates at $\{A,B\}$, for some partition $A, B$ such that $A\cap S \ne \emptyset$ and $B\cap S\ne \emptyset$.  
If $\ket{\psi}$ is not $S$-separable, then  $\ket{\psi}$ is \emph{$S$-entangled}.
\end{defn}

Observe that separation at $\{A,B\}$ is not affected by gates that act on qubits entirely within one of the sets $A$ or $B$: If $\ket{\psi}$ separates at $\{A,B\}$ and $U$ is a gate touching only qubits in $A$, say, then $U\ket{\psi}$ separates at $\{A,B\}$.  If follows that such gates do not affect $S$-separability.

\begin{defn}[Simplification of states]\label{def:simplify}
Suppose we have an $r$-qubit register with qubits labeled $1,\ldots,r$, a set $S\subseteq [r]$, and an $r$-qubit state $\ket{\psi}$.  
\begin{enumerate}\alphaitems
\item
Gate $\ctrlZ{S}$ \emph{disappears on} 
(or \emph{is turned off by}) $\ket{\psi}$,
if $\ctrlZ{S}\ket{\psi} = \ket{\psi}$.
\item
Gate \emph{$\ctrlZ{S}$ simplifies to $\ctrlZ{T}$ on $\ket{\psi}$},
if $\ctrlZ{S}\ket{\psi} = \ctrlZ{T}\ket{\psi} \ne \ket{\psi}$, for some 
$T\subsetneq S$.  
\end{enumerate}
We say that $\ctrlZ{S}$ \emph{simplifies} on $\ket{\psi}$ if either~(a) or~(b) hold.
\obs{
We say that $\ctrlZ{S}$ \emph{simplifies} on $\ket{\psi}$ if either (a) $\ctrlZ{S}\ket{\psi} = \ket{\psi}$ or (b) $\ctrlZ{S}\ket{\psi} = \ctrlZ{T}\ket{\psi} \ne \ket{\psi}$ for some proper subset $T\subset S$.  
In case (a), we say that $\ctrlZ{S}$ \emph{disappears on} (or \emph{is turned off by}) $\ket{\psi}$; in case (b), we say that \emph{$\ctrlZ{S}$ simplifies to $\ctrlZ{T}$ on $\ket{\psi}$}.
}
\end{defn}

Observe that the two cases (a) and (b) in Definition~\ref{def:simplify} above are mutually exclusive, given $S$ and $\ket{\psi}$.  Also observe that $\ctrlZ{S}$ disappears on $\ket{\psi}$ if and only if $\braket{x}{\psi} = 0$ for every computational basis state $\ket{x}$ such that the string $x$ has $1$'s in all positions in $S$.  $\ctrlZ{S}$ simplifies to $\ctrlZ{T}$ on $\ket{\psi}$ if and only if $\braket{x}{\psi} = 0$ for every computational basis state $\ket{x}$ where $x$ has a $0$ in some position in $S-T$; equivalently, $\ket{\psi}$ factors into a tensor product of a $\ket{1}$ state of each qubit in $S-T$, along with some arbitrary state of the rest of the qubits.

We will use Lemmas~\ref{lem:most-general}, \ref{lem:most-general-one-zero}, and \ref{lem:most-general-two-zeros} to prove the next lemma, which is the main lemma of this section.

\begin{lemma}[Entanglement Lemma]\label{lem:S-entangled}
Suppose we have an $r$-qubit register as in Definition~\ref{def:simplify}, and let $S \subseteq[r]$.  Let $\ket{\psi}$ be any state of the register, and let $\ket{\p} := \ctrlZ{S}\ket{\psi}$.  Then at least one of the following must hold: 
\begin{enumerate}
\item
$\ket{\psi}$ is $S$-entangled,
\item 
$\ket{\p}$ is $S$-entangled,
\item 
$\ctrlZ{S}$ simplifies on $\ket{\psi}$.
\end{enumerate}
\end{lemma}

\begin{proof}
Let $\{A,B\}$ and $\{C,D\}$ be two partitions of~$[r]$
such that all four sets have nonempty intersection with~$S$.
Let $\ket{\psi}_A$ and $\ket{\psi}_B$ be arbitrary states of the qubits in $A$ and $B$, respectively, and let $\ket{\psi} := \ket{\psi}_A \otimes \ket{\psi}_B$.  Define $\ket{\p} := \ctrlZ{S}\ket{\psi}$.  Suppose that $\ctrlZ{S}$ does not simplify on $\ket{\psi}$.  We will show that $\ket{\p}$ cannot be written as a tensor product of two states---one on the qubits in $C$ and the other on the qubits in $D$.  As $C$ and $D$ were chosen arbitrarily, this shows that $\ket{\p}$ is $S$-entangled, hence the lemma follows.

The two partitions $\{A,B\}$ and $\{C,D\}$ lead to a $4$-partition of $[r]$ into sets~$A\cap C$, $A\cap D$, $B\cap C$, and $B\cap D$, some of which may be empty.  By rearranging qubits, we may assume WLOG that for some $k,\ell,m\in\nums$ we have
\begin{align}
A\cap C &= \{1,\ldots,k\}\;, & A\cap D &= \{k+1,\ldots,k+\ell\}\;, \\
B\cap C &= \{k+\ell+1,\ldots,k+\ell+m\}\;, & B\cap D &= \{k+\ell+m+1,\ldots,r\}\;.
\end{align}
Setting $n := r - k - \ell - m$, we then have
\begin{align}
|A\cap C| &= k\;, & |A\cap D| &= \ell\;, & |B\cap C| &= m\;, & |B\cap D| &= n\;.
\end{align}
By rearranging the qubits within these four sets if necessary, we may also assume that their intersections with $S$ occur \emph{first} within each set.  For example, $A\cap C\cap S = \{1,\ldots,k_1\}$ for some $0\le k_1\le k$, and we set $k_2 := k - k_1$; similarly for the other three sets.  The full layout is shown in Figure~\ref{fig:ABCD-layout}.
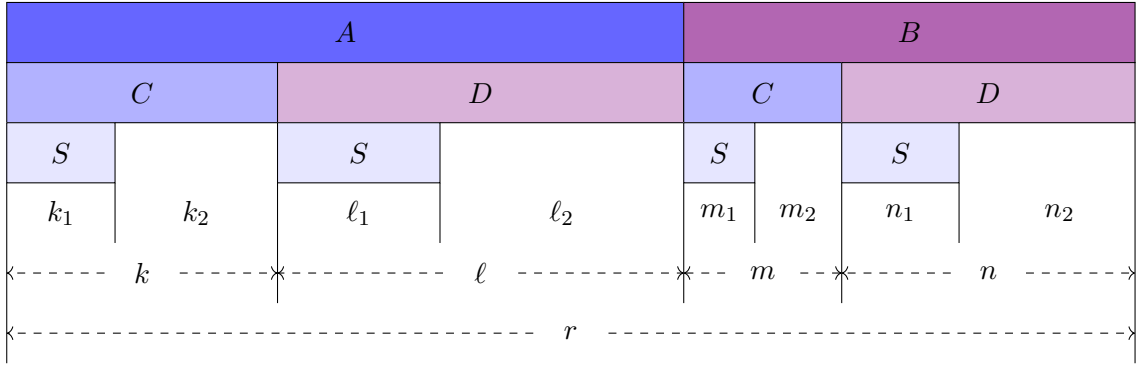
\begin{figure}
\begin{center}
\begin{tikzpicture}
\def\w{15}         
\def\h{0.8}        

\def\fa{0.6}       
\def\fca{0.4}      
\def\fcb{0.35}     
\def\fsca{0.4}      
\def\fsda{0.4}   
\def\fscb{0.45}   
\def\fsdb{0.4}   


\def\a{\fa*\w}     
\def\ca{\fca*\a}    
\def\cb{\fcb*\w-\fcb*\a}   
\def\sca{\fsca*\ca}  
\def\sda{\fsda*\a-\fsda*\ca}  
\def\scb{\fscb*\fcb*\w-\fscb*\fcb*\a} 
\def\sdb{\fsdb*\w -\fsdb*\a-\fsdb*\fcb*\w+\fsdb*\fcb*\a} 

\node (A1) at (0,-\h) {};
\node (A2) at (\a,0) {};

\node (B1) at (\a,-\h) {};
\node (B2) at (\w,0) {};


\node (CA1) at (0,-2*\h) {};
\node (CA2) at (\ca,-\h) {};

\node (DA1) at (\ca,-2*\h) {};
\node (DA2) at (\a,-\h) {};

\node (CB1) at (\a,-2*\h) {};
\node (CB2) at (\a+\cb,-\h) {};

\node (DB1) at (\a+\cb,-2*\h) {};
\node (DB2) at (\w,-\h) {};


\node (SCA1) at (0,-3*\h) {};
\node (SCA2) at (\sca,-2*\h) {};

\node (SDA1) at (\ca,-3*\h) {};
\node (SDA2) at (\ca+\sda,-2*\h) {};

\node (SCB1) at (\a,-3*\h) {};
\node (SCB2) at (\a+\scb,-2*\h) {};

\node (SDB1) at (\a+\cb,-3*\h) {};
\node (SDB2) at (\a+\cb+\sdb,-2*\h) {};


\draw[fill=blue!60] (A1) rectangle (A2) node[midway] {$A$};
\draw[fill=violet!60] (B1) rectangle (B2) node[midway] {$B$};

\draw[fill=blue!30] (CA1) rectangle (CA2) node[midway] {$C$};
\draw[fill=violet!30] (DA1) rectangle (DA2) node[midway] {$D$};
\draw[fill=blue!30] (CB1) rectangle (CB2) node[midway] {$C$};
\draw[fill=violet!30] (DB1) rectangle (DB2) node[midway] {$D$};

\draw[fill=blue!10] (SCA1) rectangle (SCA2) node[midway] {$S$};
\draw[fill=blue!10] (SDA1) rectangle (SDA2) node[midway] {$S$};
\draw[fill=blue!10] (SCB1) rectangle (SCB2) node[midway] {$S$};
\draw[fill=blue!10] (SDB1) rectangle (SDB2) node[midway] {$S$};


\draw
(0,-6*\h) -- (0,-3*\h)
(\sca,-4*\h) -- (\sca,-3*\h)
(\ca,-5*\h) -- (\ca,-3*\h)
(\ca+\sda,-4*\h) -- (\ca+\sda,-3*\h)
(\a,-5*\h) -- (\a,-3*\h)
(\a+\scb,-4*\h) -- (\a+\scb,-3*\h)
(\a+\cb,-5*\h) -- (\a+\cb,-3*\h)
(\a+\cb+\sdb,-4*\h) -- (\a+\cb+\sdb,-3*\h)
(\w,-6*\h) -- (\w,-2*\h)
;


\def\xk{0.5*\ca}

\node at (0.5*\sca,-3.5*\h) {$k_1$};
\node at (0.5*\ca+0.5*\sca,-3.5*\h) {$k_2$};
\node at (\xk,-4.5*\h) {$k$};
\draw
(\xk-0.5,-4.5*\h) edge[dashed,->] node[above,draw=none]{} (0,-4.5*\h)
(\xk+0.5,-4.5*\h) edge[dashed,->] node[above,draw=none]{} (\ca,-4.5*\h)
;

\def\xl{\ca+0.5*\a-0.5*\fca*\a}

\node at (\ca+0.5*\fsda*\a-0.5*\fsda*\ca,-3.5*\h) {$\ell_1$};
\node at (\ca+0.5*\a-0.5*\fca*\a+0.5*\fsda*\a-0.5*\fsda*\ca,-3.5*\h) {$\ell_2$};
\node at (\xl,-4.5*\h) {$\ell$};
\draw
(\xl-0.5,-4.5*\h) edge[dashed,->] node[above,draw=none]{} (\ca,-4.5*\h)
(\xl+0.5,-4.5*\h) edge[dashed,->] node[above,draw=none]{} (\a,-4.5*\h)
;

\def\xm{\a+0.5*\fcb*\w-0.5*\fcb*\a}

\node at (\a+0.5*\fscb*\fcb*\w-0.5*\fscb*\fcb*\a,-3.5*\h) {$m_1$};
\node at (\a+0.5*\fcb*\w-0.5*\fcb*\a+0.5*\fscb*\fcb*\w-0.5*\fscb*\fcb*\a,-3.5*\h) {$m_2$};
\node at (\xm,-4.5*\h) {$m$};
\draw
(\xm-0.5,-4.5*\h) edge[dashed,->] node[above,draw=none]{} (\a,-4.5*\h)
(\xm+0.5,-4.5*\h) edge[dashed,->] node[above,draw=none]{} (\a+\cb,-4.5*\h)
;

\def\xn{\a+\cb+0.5*\w-0.5*\a-0.5*\fcb*\w+0.5*\fcb*\a}

\node at (\a+\cb+0.5*\fsdb*\w -0.5*\fsdb*\a-0.5*\fsdb*\fcb*\w+0.5*\fsdb*\fcb*\a,-3.5*\h) {$n_1$};
\node at (\a+\cb+\scb+0.5*\w-0.5*\a-0.5*\fcb*\w+0.5*\fcb*\a,-3.5*\h) {$n_2$};
\node at (\xn,-4.5*\h) {$n$};
\draw
(\xn-0.5,-4.5*\h) edge[dashed,->] node[above,draw=none]{} (\a+\cb,-4.5*\h)
(\xn+0.5,-4.5*\h) edge[dashed,->] node[above,draw=none]{} (\w,-4.5*\h)
;

\node at (0.5*\w,-5.5*\h) {$r$};
\draw
(0.5*\w-0.5,-5.5*\h) edge[dashed,->] node[above,draw=none]{} (0,-5.5*\h)
(0.5*\w+0.5,-5.5*\h) edge[dashed,->] node[above,draw=none]{} (\w,-5.5*\h)
;

\end{tikzpicture}
\end{center}
\caption{The most general partitioning of $[r]$ into intersections of the sets $A,B,C,D,S$.  Some intersections may be empty.}\label{fig:ABCD-layout}
\end{figure}

The constraint that each of $A,B,C,D$ intersects $S$ implies that the quantities $k_1+\ell_1$, $m_1+n_1$, $k_1+m_1$, and $\ell_1+n_1$ are all positive.  By swapping the roles of $C$ and $D$ if necessary, we may further assume that $k_1$ and $n_1$ are both positive.

We now consider four cases: (1) $\ell_1>0$ and $m_1>0$; (2) $\ell_1=0$ and $m_1>0$; (3) $\ell_1>0$ and $m_1=0$; (4) $\ell_1=m_1=0$.  Cases~(2) and (3) are essentially the same case, because one can be converted to the other by simultaneously swapping the roles of $A$ and $B$ and swapping the roles of $C$ and $D$.  Thus we can ignore Case~(3) without loss of generality.

\paragraph{Case~(1):}  In this case, $A\cap C\cap S$, $A\cap D\cap S$, $B\cap C\cap S$, and $B\cap D\cap S$ are all nonempty.  Let $\cH_{AC},\cH_{AD},\cH_{BC},\cH_{BD}$ be the spaces on qubits in $A\cap C,A\cap D,B\cap C,B\cap D$, respectively.  Then $\ket{\psi}_A \in \cH_{AC}\otimes\cH_{AD}$ and $\ket{\psi}_B \in \cH_{BC}\otimes\cH_{BD}$.  We can then write $\ket{\psi}_A$ and $\ket{\psi}_B$ uniquely as
\begin{align}
\ket{\psi}_A &= \sum_{s\in\two^k}\sum_{t\in\two^{\ell}} c_{s,t}\ket{s}\otimes\ket{t}\;, & \ket{\psi}_B &= \sum_{u\in\two^m}\sum_{v\in\two^n} d_{u,v}\ket{u}\otimes\ket{v}\;,
\end{align}
where the $c_{s,t}$ and $d_{u,v}$ are coefficients in $\complexes$.  Then using Notation~\ref{not:concat-convention} conventions,
\begin{align}
\ket{\psi} &= \ket{\psi}_A\otimes\ket{\psi}_B = \sum_{s,t,u,v} c_{s,t} d_{u,v} \ket{s}\otimes\ket{t}\otimes\ket{u}\otimes\ket{v} \\
&=\sum_{s,t,u,v} c_{s,t} d_{u,v} \ket{s_1\circ s_2}\otimes\ket{t_1\circ t_2}\otimes\ket{u_1\circ u_2}\otimes\ket{v_1\circ v_2}\;.\label{eqn:ket-psi}
\end{align}
Applying $\ctrlZ{S}$ to $\ket{\psi}$ flips the sign of each term where $s_1,t_1,u_1,v_1$ (corresponding to the positions in $S$) are all $\ones$'s.  Thus
\begin{align}
\ket{\p} &= \ctrlZ{S}\ket{\psi} = \sum_{s,t,u,v} c_{s,t} d_{u,v} \ket{s_1\circ s_2}\otimes\ket{t_1\circ t_2}\otimes\ket{u_1\circ u_2}\otimes\ket{v_1\circ v_2} \nonumber\\
&\mbox{\hspace{1in}} - 2\sum_{s_2,t_2,u_2,v_2} c_{s,t} d_{u,v} \ket{\ones\circ s_2}\otimes\ket{\ones\circ t_2}\otimes\ket{\ones\circ u_2}\otimes\ket{\ones\circ v_2} \label{eqn:ket-phi}
\end{align}

Using Eq.~(\ref{eqn:poly-four}), we then get
\[ \poly_{\cH_{AC},\cH_{AD},\cH_{BC},\cH_{BD}}(\ket{\psi}) = \sum_{s,t,u,v} c_{s,t} d_{u,v}\, x_sy_tz_uw_v = T_1T_2\;, \]
where
\begin{align}
T_1 &= \sum_{s,t}c_{s,t}x_sy_t\;, & T_2 &= \sum_{u,v} d_{u,v}z_uw_v\;.
\end{align}
Thus
\[ \poly_{\cH_{AC},\cH_{AD},\cH_{BC},\cH_{BD}}(\ket{\p}) = T_1T_2 - 2\sum_{s_2,t_2,u_2,v_2} c_{s,t} d_{u,v}\,x_{\ones\circ s_2}y_{\ones\circ t_2}z_{\ones\circ u_2}w_{\ones\circ v_2} = P\;, \]
where $P$ is given by Eq.~(\ref{eqn:C-no-zeros}) of Lemma~\ref{lem:most-general} with $\alpha = 2$.  Assuming the hypotheses of that lemma hold, $P$ is irreducible.  One cannot write $\ket{\p}$ as a tensor product $\ket{\p}_C\otimes\ket{\p}_D$ then, for otherwise, $P = \poly_{\cH_{AC},\cH_{BC}}(\ket{\p}_C)\cdot\poly_{\cH_{AD},\cH_{BD}}(\ket{\p}_D)$ with each factor being nonconstant, contradicting the irreducibility of $P$.  Since $C$ and $D$ were chosen arbitrarily subject to the constraints of Case~(1), it follows that $\ket{\p}$ is $S$-entangled.  It remains to show that the hypotheses of Lemma~\ref{lem:most-general} hold in this case.

Our assumption that $\ctrlZ{S}$ does not simplify on $\ket{\psi}$ puts constraints on the coefficients $c_{s,t}$ and $d_{u,v}$.  Since $\ctrlZ{S}$ does not disappear, the expression for $\ket{\psi}$ in Eq.~(\ref{eqn:ket-psi}) must include at least one term in the sum with all $1$'s being fed into $\ctrlZ{S}$, that is, there is some nonzero term of the form
\[ c_{s,t}d_{u,v}\ket{\ones\circ s_2}\otimes\ket{\ones\circ t_2}\otimes\ket{\ones\circ u_2}\otimes\ket{\ones\circ v_2}\;. \]
Thus $c_{s,t}\ne 0$ and $d_{u,v} \ne 0$ for this choice of $s,t,u,v$.
This matches the hypotheses $2(a,b)$ of Lemma~\ref{lem:most-general}.

We also assume that $\ctrlZ{S}$ does not simplify to $\ctrlZ{T}$ on $\ket{\psi}$ for any  $T\subsetneq S$.  Such a simplification occurs when there is some position $p\in S$ such that $\ket{\psi}$ factors into a state with $\ket{1}$ on qubit~$p$, unentangled with the state of the other qubits, in which case we have $T\subseteq S\setminus\{p\}$.  For this \emph{not} to happen then, for every $p\in S$, there is a nonzero term in the sum of Eq.~(\ref{eqn:ket-psi}) whose basis state $\ket{s}\otimes\ket{t}\otimes\ket{u}\otimes\ket{v}$ has $0$ in position~$p$.  Since $S$ has nonempty intersection with all four sets $A\cap C, A\cap D, B\cap C, B\cap D$ (because we are in Case~(1)), hypotheses $2(c,d)$ of Lemma~\ref{lem:most-general} must hold.  This concludes the proof for Case~(1).

Cases~(2) and (4) are simpler but completely analogous to Case~(1).  Instead of using Lemma~\ref{lem:most-general}, Case~(2) uses Lemma~\ref{lem:most-general-one-zero} and Case~(4) uses Lemma~\ref{lem:most-general-two-zeros}.  We omit the details.
\end{proof}

Rather than using Lemma~\ref{lem:S-entangled} directly, we will use the following stronger corollary.

\begin{lemma}\label{lem:simplify}
Let $r$ and $S\subseteq [r]$ be as in Lemma~\ref{lem:S-entangled}, and let $\{A,B\}$ and $\{C,D\}$ be two partitions of~$[r]$.  Let $\ket{\psi}_A,\ket{\psi}_B,\ket{\p}_C,\ket{\p}_D$ be states in $\cH_A,\cH_B,\cH_C,\cH_D$, respectively.  If $\ctrlZ{S}(\ket{\psi}_A\otimes\ket{\psi}_B) = \ket{\p}_C\otimes\ket{\p}_D$, then either $\ctrlZ{S}$ disappears on $\ket{\psi}_A\otimes\ket{\psi}_B$ or $\ctrlZ{S}$ simplifies to $\ctrlZ{T}$ on $\ket{\psi}_A\otimes\ket{\psi}_B$, where $T\subseteq S$ is a subset of one of the sets $A,B,C,D$.
\end{lemma}

\begin{proof}
Let $\ket{\psi} := \ket{\psi}_A\otimes\ket{\psi}_B$ and $\ket{\p} := \ket{\p}_C\otimes\ket{\p}_D$.  Suppose $\ctrlZ{S}$ does not disappear on $\ket{\psi}$.  If $S$ is a subset of one of $A,B,C,D$, then we can set $T := S$ and we are done.  Otherwise, both $\ket{\psi}$ and $\ket{\p}$ are $S$-separable.  Therefore by Lemma~\ref{lem:S-entangled}, $\ctrlZ{S}$ simplifies to $\ctrlZ{T_1}$ on $\ket{\psi}$ for some  $T_1\subsetneq S$.  If $T_1$ is a subset of one of $A,B,C,D$, then we are done.  Otherwise, $\ket{\psi}$ and $\ket{\p}$ are both $T_1$-separable, and applying Lemma~\ref{lem:S-entangled} again, we get that $\ctrlZ{T_1}$ (hence also $\ctrlZ{S}$) simplifies on $\ket{\psi}$ to $\ctrlZ{T_2}$ for some $T_2\subsetneq T_1$, etc., eventually getting $\ctrlZ{S}$ to simplify to $\ctrlZ{T_j}$ on $\ket{\psi}$ for some $j$ such that $T_j$ is a subset of one of $A,B,C,D$.  Set $T := T_j$.
\end{proof}

\begin{rmrk}
Lemmas~\ref{lem:S-entangled} and~\ref{lem:simplify} hold not just for a $\ctrlZ{}$-gate but for any $r$-qubit gate $G_\eta$ defined for all $x = x_1\cdots x_r\in\two^r$ as
\begin{equation}\label{eqn:G-eta}
G_\eta\ket{x} := \begin{cases}
\eta\ket{x}, & \mbox{if $x = 1^r$,} \\
\ket{x},     & \mbox{otherwise,}
\end{cases}
\end{equation}
where $\eta\in\complexes$ satisfies $|\eta| = 1$ and $\eta \ne 1$.  One just replaces the ``$\mbox{}-2$'' in Eq.~(\ref{eqn:ket-phi}) with ``$\mbox{}+(\eta-1)$.''
\end{rmrk}

We will also need the next routine lemma, which says that if $\ctrlZ{S}$ disappears on some tensor product state $\ket{\psi}_A\otimes\ket{\psi}_B$, then one of the states $\ket{\psi}_A$ or $\ket{\psi}_B$ completely ensures that $\ctrlZ{S}$ disappears.  More precisely, we have the following:

\begin{lemma}\label{lem:no-zero-divisors}
Let $r$ and $S\subseteq[r]$ be as in Lemma~\ref{lem:S-entangled} and let $\{A,B\}$ be a partition of~$[r]$
Suppose $\ctrlZ{S}$ disappears on $\ket{\psi}_A\otimes\ket{\psi}_B$, for some states $\ket{\psi}_A\in\cH_A$ and $\ket{\psi}_B\in\cH_B$.  Then either $\ctrlZ{S}$ disappears on $\ket{\psi}_A\otimes\ket{\p}_B$ for all states $\ket{\p}_B\in\cH_B$, or $\ctrlZ{S}$ disappears on $\ket{\p}_A\otimes\ket{\psi}_B$ for all states $\ket{\p}_A\in\cH_A$.
\end{lemma}

\begin{proof}
Let $\cH_{\ones,A}$, $\cH_{\ones,B}$, and $\cH_{\ones}$ be the subspaces of $\cH_A$, $\cH_B$, and $\cH_r = \cH_{A\cup B}$, respectively, that are spanned by those basis vectors corresponding to strings with $1$'s in all the positions in $S\cap A$, $S\cap B$, and $S$, respectively.  We can write
\begin{align}
\ket{\psi}_A &= \alpha\ket{\ones}_A + \beta\ket{\ones_\perp}_A\;, \\
\ket{\psi}_B &= \gamma\ket{\ones}_B + \delta\ket{\ones_\perp}_B\;,
\end{align}
where $\alpha,\beta,\gamma,\delta\in\complexes$ and $\ket{\ones}_A$ is a unit vector in $\cH_{\ones,A}$ and $\ket{\ones_\perp}_A$ is unit vector in the orthogonal complement of $\cH_{\ones,A}$ in $\cH_A$ (spanned by the basis states that include at least one $0$ in a position in $S\cap A$).  Similarly for $\ket{\ones}_B$ and $\ket{\ones_\perp}_B$.  We then have
\[ \ket{\psi}_A\otimes\ket{\psi}_B = \alpha\gamma\ket{\ones} + u\;, \]
where $\ket{\ones}$ is a unit vector in $\cH_{\ones}$ and $u$ is some vector in its orthogonal complement $\cH_{\ones}^{\perp}$.  $\ket{\psi}_A\otimes\ket{\psi}_B$ turns off $\ctrlZ{S}$ if and only if it is in $\cH_{\ones}^{\perp}$, i.e., iff $\alpha\gamma = 0$.  If $\alpha = 0$, then $\ket{\psi}_A = \ket{\ones_\perp}_A$ up to a phase factor, which implies $\ket{\psi}_A\otimes\ket{\p}_B \in \cH_{\ones}^{\perp}$ for any $\ket{\p}_B\in\cH_B$.  Similarly if $\gamma = 0$.
\end{proof}

\section{Pure Parity States}

\begin{defn}[Subspace $\calP_b$]\label{def:parity-space}
Given $r\ge 1$ and $b\in\binary$, 
we define the subspace $\calP_b$ of $\cH_r$ to be the space spanned by $\set{\ket{x} \st x\in\bits{r} \wedge \oplus x = b}$.
\end{defn}

Clearly, $\dim\calP_0 = \dim\calP_1 = 2^{r-1}$, and $\cH_r$ is the direct sum of $\calP_0$ and $\calP_1$.

\begin{defn}[Parity of a State]
Given an $r$-qubit state $\ket{\psi}\in\cH_r$ and $b\in\binary$,
we say that \emph{$\ket{\psi}$ has pure parity~$b$} if $\ket{\psi}\in\calP_b$.  We say that $\ket{\psi}$ is a \emph{pure parity state} if $\ket{\psi}$ has pure parity~$b$ for some $b\in\binary$.
\end{defn}

Every classical state (i.e., computational basis state) is clearly a pure parity state, and the tensor product of pure parity states on disjoint sets of qubits is itself a pure parity state.  If a quantum circuit $C$ weakly computes $\oplus_n$ for some $n$, witnessed by an initial state $\ket{\psi}$ of the ancilla qubits (cf.\ Definition~\ref{def:compute-f}), then setting the input qubits to a state with pure parity $b$ must result in an output of the form $\ket{b}\otimes\ket{\p}$ for some state $\ket{\p}$ of the non-target qubits ($\ket{\p}$ may depend on the state of the input qubits).  In particular, the final state separates at $\{\{0\},\overline{\{0\}}\}$.

\begin{lemma}\label{lem:kill-parity}
Given any $r$-qubit unitary operators $U_1,\ldots,U_k$ for some $k < 2^{r-1}$ and any bit $b\in\binary$, there is an $r$-qubit state $\ket{\psi_b}$ with pure parity~$b$ such that $\bra{1^r} U_i \ket{\psi_b} = 0$ for all $1\le i\le k$.
\end{lemma}

\begin{proof}
Let $\calP_0$ and $\calP_1$ be as in Definition~\ref{def:parity-space}.  For $1\le i\le k$, let $\cZ_i\subseteq\cH_r$ be the $(2^r-1)$-dimensional subspace of $\cH_r$ spanned by $\set{U_i^*\ket{x}: x\in\binary^r\setminus \{1^r\}}$.  Then for all $i$, \ $\bra{1^r} U_i \ket{\psi} = 0$
for any state $\ket{\psi}\in\cZ_i$.  Letting $\cZ := \bigcap_{i=1}^k \cZ_i$, we see that $\dim(\cZ) \ge 2^r-k$.  For $b\in\binary$, we then have
\[ \dim(\calP_b\cap\cZ) = \dim\calP_b + \dim\cZ - \dim(\calP_b+\cZ) \ge \dim\calP_b + \dim\cZ - 2^r \ge 2^{r-1} + (2^r - k) - 2^r \ge 1\;. \]
It follows that we can choose a state (unit vector) $\ket{\psi_b}$ in $\calP_b\cap\cZ$, and this vector has the desired properties.
\end{proof}

We will use Lemma~\ref{lem:kill-parity} to turn off $\ctrlZ{}$-gates.  If some $\ctrlZ{}$-gate~$G$ that touches all $r$ qubits (and possibly other qubits) is applied to $U_i\ket{\psi_b}\otimes\cdots$, then $G$ is turned off, i.e., $G(U_i\ket{\psi_b}\otimes\cdots) = U_i\ket{\psi_b}\otimes\cdots$, where ``$\cdots$'' represents some state of the other qubits, if they are present.

\section{Quantum Circuit Lower Bounds}
\label{sec:lower-bounds}

In this section we prove 
 that no depth-2 $\QAC$-circuit computes~$\oplus_n$ for $n>3$ (see Definition~\ref{def:compute-f}),
which improves upon a previous version of our paper~\cite{PFGT:fanout}.

\begin{thm}\label{thm:depth-2-unclean}
No depth-$2$ $\QAC$-circuit computes $\oplus_n$ for $n>3$.
\end{thm}

This result is tight in the sense that there is a simple $4$-qubit depth-$2$ $\QAC$-circuit that computes~$\oplus_3$:
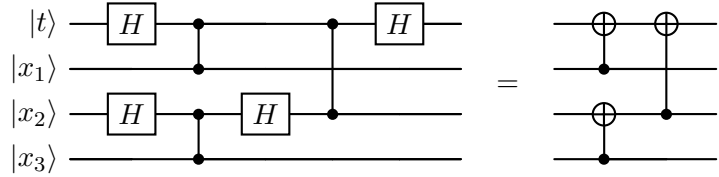
\begin{center}
\begin{quantikz}[row sep={0.6cm,between origins}]
\lstick{$\ket{t}$}   & \gate{H} & \control{} & \qw & \control{} & \gate{H} &\qw
\\
\lstick{$\ket{x_1}$} & \qw      & \ctrl{-1}  & \qw & \qw        & \qw      &\qw \\
\lstick{$\ket{x_2}$} & \gate{H} & \control{}&\gate{H}&\ctrl{-2} & \qw      &\qw \\
\lstick{$\ket{x_3}$} & \qw      & \ctrl{-1}  & \qw & \qw        & \qw      &\qw
\end{quantikz}
\;\;=\;\;
\begin{quantikz}[row sep={0.6cm,between origins}]
& \targ{}   & \targ{}   & \ghost{H} \\
& \ctrl{-1} & \qw       & \\
& \targ{}   & \ctrl{-2} & \\
& \ctrl{-1} & \qw       &
\end{quantikz}
\end{center}

We will prove Theorem~\ref{thm:depth-2-unclean} through a sequence of lemmas that may be useful in proving lower bounds for circuits of higher depth.  We will also make repeated use of the Entanglement Lemma (Lemma~\ref{lem:S-entangled}) and Lemma~\ref{lem:kill-parity}.  We adopt the conventions of Definition~\ref{def:gate-positions} to describe gates within circuits.

\begin{lemma}\label{lem:single-layer}
There is no depth-$1$ $\QAC$-circuit that weakly computes~$\parity_n$ for $n\ge 2$.
\end{lemma}

\begin{proof}
Consider such a circuit $C$ on $n\ge 2$ input qubits, witnessed by some fixed initial state of the ancilla qubits.  The target and first two input qubits must all be incident to a single gate $G_0^{(1)} = \ctrlZ{S}$ for some $S\supseteq\{0,1,2\}$, for otherwise there is an input qubit that does not interact with the target qubit at all, whence $C$ cannot weakly compute $\parity_n$.  Then by Lemma~\ref{lem:kill-parity} (with $r:=2$ and $U_1 := G_{\{1,2\}}^{(1)}$), for each $b\in\two$, qubits~$1$ and $2$ can be initially committed to a $2$-qubit state with pure parity $b$ that turns off $G_0^{(1)}$.  With either of these initial states (setting any other input qubits to $\ket{0}$), the target does not interact with any other qubits and so can only be $G_0^{(1.5)}G_0^{(0.5)}\ket{0}$.  But then the final state of the target does not depend on $b$, and thus $C$ does not weakly compute $\oplus_n$.
\end{proof}

\begin{defn}
We will say that a $1$-qubit gate $U$ is \emph{semiclassical} if its $2\times 2$ matrix representation with respect to the computational basis has two entries that are~$0$.  Equivalently, $U\ket{0}$ is a computational basis state up to a phase factor.
\end{defn}

Observe that a $1$-qubit unitary gate $U$ is semiclassical if and only if $U^*$ is semiclassical.

\begin{defn}
In a depth-$d$ $\QAC$-circuit, if the $1$-qubit gate $G_0^{(d+1/2)}$ in layer~$d+\frac{1}{2}$ of the target is semiclassical, then we say that the target is \emph{pass-through}.
\end{defn}

\begin{lemma}\label{lem:target-cant-pass}
For any $n\ge 1$ and $d\ge 2$, let $C$ be a depth-$d$ $\QAC$-circuit that $\ket{\alpha}$-computes $\parity_n$, for some state $\ket{\alpha}$.  If $C$'s target is either pass-through or does not encounter a multiqubit $\ctrlZ{}$-gate on layer~$d$, then there exists a depth-$(d-1)$ $\QAC$-circuit that $\ket{\alpha}$-computes $\parity_n$.
\end{lemma}

\begin{proof}
Fix an initial ancilla state $\ket{\alpha}$ that witnesses $C$ $\ket{\alpha}$-computing $\parity_n$.
For any classical input $x$ combined with $\ket{\alpha}$, the final state of the target (qubit~$0$ after layer~$d+1/2$) is $\ket{b}$ unentangled with any other qubits, where $b := \oplus x$.  We have two cases:

\paragraph{Case~1: $C$'s target does not encounter a multiqubit $\ctrlZ{}$-gate on layer $d$.}  $G_0^{(d)}$ is then a $1$-qubit gate---either $I$ or $Z$.  Thus the final target state is not affected by any other non-target gates beyond layer $d-1$.  Let $C'$ be the depth~$(d-1)$ circuit obtained by removing all these gates and collapsing $G_0^{(d-1/2)}$, $G_0^{(d)}$, and $G_0^{(d+1/2)}$ into a single gate.  The final state of the target is thus the same with $C'$ as with $C$, and so $C'$ $\ket{\alpha}$-computes $\oplus_n$.

\paragraph{Case~2: $C$'s target is pass-through.}
Since qubit~$0$ is pass-through by assumption, the target just after layer~$d$ is in an unentangled computational basis state $\ket{\p_b}$ that equals $\left(G_0^{(d+1/2)}\right)^*\ket{b}$ up to a phase factor (which can be absorbed by the state of the other qubits).  Thus if $G_0^{(d)}$ is a multiqubit $\ctrlZ{}$-gate, it either disappears or simplifies to a $\ctrlZ{}$-gate not acting on the target, depending on $b$.  In either case, the (unentangled) state of the target is unchanged across layer~$d$.  Let $C'$ be the depth-$(d-1)$ circuit obtained from $C$ by removing all gates on layer~$d$, removing all non-target gates on layer~$d+1/2$, and combining $G_0^{(d+1/2)}$ with $G_0^{(d-1/2)}$.  Then $C'$ $\ket{\alpha}$-computes $\parity_n$.
\end{proof}

The following lemma is a corollary to Lemma~\ref{lem:target-cant-pass}.

\begin{lemma}\label{lem:target-cant-pass-2}
In any depth-$2$ $\QAC$-circuit weakly computing $\parity_n$ for $n\ge 2$, \ $G_0^{(2)}$ is a multiqubit $\ctrlZ{}$-gate, and the target is not pass-through.
\end{lemma}

\begin{proof}
By Lemmas~\ref{lem:single-layer} and \ref{lem:target-cant-pass}.
\end{proof}

In the sequel, we assume that $C$ is an $(n+m+1)$-qubit depth-$2$ $\QAC$-circuit weakly computing~$\oplus_n$ for some $n\ge 3$ (cf.\ Definition~\ref{def:compute-f}).  By Lemma~\ref{lem:target-cant-pass-2}, $G_0^{(2)} = \ctrlZ{S}$ for some set $S$ that includes the target and at least one other qubit, and the target is not pass-through.  The next few lemmas restrict the topology of $C$ further.

\begin{lemma}\label{lem:no-three-input-qubits}
No gate on layer~$1$ can touch more than two input qubits.
\end{lemma}

\begin{proof}
Suppose some layer~$1$ gate touches at least three input qubits.  WLOG, $G_1^{(1)} = \ctrlZ{T}$ for some $T$ such that $\{1,2,3\}\subseteq T$.  We let $\ket{\alpha}$ be the initial state of the $m$ ancilla qubits.  We consider two cases and apply Lemma~\ref{lem:kill-parity} to each:

\paragraph{Case~1: $G_0^{(2)}$ does not touch one of the qubits~$1$, $2$, or $3$.}  WLOG, $3\notin S$.  By Lemma~\ref{lem:kill-parity} (with $r := 2$ and $U_1 := G_{\{1,2\}}^{(0.5)}$), we can choose an initial pure parity state $\ket{\psi}\in\cH_{\{1,2\}}$ (of pure parity~$0$, say) of qubits~$1$ and $2$ that turns $G_1^{(1)}$ off, regardless of the initial state of the other qubits.  But then, qubit~$3$ has no connection to the target at all, and so the final state of the target is independent of the third input bit, regardless of the rest of the input bits and the initial state of the ancilla.  Particularly, for any $b\in\two$, let the initial state of the circuit be
\[ \ket{0}\otimes\ket{\psi}\otimes\ket{b}\otimes\ket{0}^{\otimes(n-3)}\otimes\ket{\alpha}\;. \]
(We set the third input qubit to $\ket{b}$ and input qubits $4,\ldots\,$, if any, to $\ket{0}$.)  Then the final state of the circuit is of the form $\ket{0}\otimes\ket{\tau_b}$, where $\ket{\tau_b}$ is the final state of the non-target qubits.  $\ket{\tau_b}$ may depend on $b$, but the final state of the target does not, and thus $C$ does not weakly compute $\otimes_n$.

\paragraph{Case~2: $G_0^{(2)}$ touches all of the qubits~$1$, $2$, and $3$ (i.e., not Case~1).}  That is, $\{1,2,3\}\subseteq S$. By Lemma~\ref{lem:kill-parity} (with $r := 3$, \ $U_1 := G_{\{1,2,3\}}^{(0.5)}$, and $U_2 := G_{\{1,2,3\}}^{(1.5)}U_1$), for each $b\in\two$ we can choose an initial state $\ket{\psi_b}\in\cH_{\{1,2,3\}}$ with pure parity~$b$ on qubits~$1$, $2$, and $3$ that turns $G_1^{(1)}$ and $G_0^{(2)}$ \emph{both} off, regardless of the initial state of the other qubits.  Thus given the initial state $\ket{0}\otimes\ket{\psi_b}\otimes\ket{0}^{\otimes(n-3)}\otimes\ket{\alpha}$, the target has no connection to the first three input qubits, so its final value cannot depend on $b$.  Since the initial state $\ket{\psi_b}\otimes\ket{0}^{\otimes(n-3)}$ of the input qubits has pure parity~$b$, \ $C$ does not weakly compute $\otimes_n$.
\end{proof}

\begin{lemma}\label{lem:target-connects-only-one-input}
$G_0^{(1)}$ can only touch at most one input qubit.
\end{lemma}

\begin{proof}
Suppose some $G_0^{(1)} = \ctrlZ{T}$, where $T$ includes the target and at least two other input qubits.  WLOG, $\{0,1,2\}\subseteq T$.  By Lemma~\ref{lem:kill-parity} (with $r:=2$ and $U_1 := G_{\{1,2\}}^{(0.5)}$), for each $b\in\two$, we can choose an initial state $\ket{\psi_b}$ of pure parity~$b$ on qubits~$1$ and $2$ that turns $G_0^{(1)}$ off, regardless of the initial state of the other qubits.  For each $b$, set the initial state of the other input qubits to all $\ket{0}$, resulting in an initial state
\[ \ket{\psi_b'} := \ket{0}\otimes\ket{\psi_b}\otimes\ket{0}^{\otimes(n-2)}\otimes\ket{\alpha} \]
where $\ket{\alpha}$ is the initial state of the ancilla qubits.  Since $\ket{\psi_b'}$ turns $G_0^{(1)}$ off, the target is not connected to any other qubits before layer~$2$.  Applying $G^{(1.5)}G^{(1)}G^{(0.5)}$ to $\ket{\psi_b'}$ thus results in a state
\[ \ket{\p_b} := \ket{\p}_{\{0\}}\otimes\ket{\p_b}_{\{1,2\}}\otimes\ket{\p}_B \]
right before layer~$2$,
where $B := \overline{\{0,1,2\}}$, \ $\ket{\p}_{\{0\}} := G_0^{(1.5)}G_0^{(0.5)}\ket{0}$ independent of $b$, \ $\ket{\p_b}_{\{1,2\}} := G_{\{1,2\}}^{(1.5)}G_{\{1,2\}}^{(0.5)}\ket{\psi_b}$, and $\ket{\p}_B$ is the state of the qubits in $B$ and is independent of $b$.  Figure~\ref{fig:simplified-circuit}
\begin{figure}
\begin{center}
\begin{quantikz}
\lstick{target $\ket{0}$}                         & \gate{G_0^{(0.5)}} & \ctrl{1}\gategroup[6,steps=1,style={dashed,rounded corners,fill=blue!20,inner xsep=2pt},background]{$G_0^{(1)}$}       & \gate{G_0^{(1.5)}} & \qw \\
\lstick[wires=2]{$\ket{\psi_b}$}           & \gate{G_1^{(0.5)}} & \ctrl{1}       & \gate{G_1^{(1.5)}} & \qw \\
                                           & \gate{G_2^{(0.5)}} & \ctrl{1}       & \gate{G_2^{(1.5)}} & \qw \\
\lstick[wires=3]{$\ket{0}^{\otimes(n-2)}\otimes\ket{\alpha}$}
                                           & \gate{G_3^{(0.5)}} & \ctrl{1}       & \gate{G_3^{(1.5)}} & \qw \\
                                           & \vdots             & \vdots         & \vdots             &     \\
                                           & \gate{G_{m+n}^{(0.5)}} & \ctrl{-1}      & \gate{G_{m+n}^{(1.5)}} & \qw
\end{quantikz}
\hspace{5mm}=\hspace{3mm}
\begin{quantikz}
& \gate{G_0^{(1.5)}G_0^{(0.5)}} & \qw\rstick{$\ket{\p}_{\{0\}}$} \\
& \gate{G_1^{(1.5)}G_1^{(0.5)}} & \qw\rstick[wires=2]{$\ket{\p_b}_{\{1,2\}}$} \\
& \gate{G_2^{(1.5)}G_2^{(0.5)}} & \qw \\
& \gate{G_3^{(1.5)}G_3^{(0.5)}} & \qw\rstick[wires=3]{$\ket{\p}_B$} \\
& \vdots                        &     \\
& \gate{G_{m+n}^{(1.5)}G_{m+n}^{(0.5)}} & \qw
\end{quantikz}
\end{center}
\caption{The portion of a typical circuit $C$ before layer~$2$.  The top line is qubit~$0$ (the target).  $\ket{\psi_b}$ on qubits~$1$ and $2$ turns $G_0^{(1)}$ off.  Here, $G_0^{(1)}$ is depicted as touching all qubits, but this need not be the case.}\label{fig:simplified-circuit}
\end{figure}
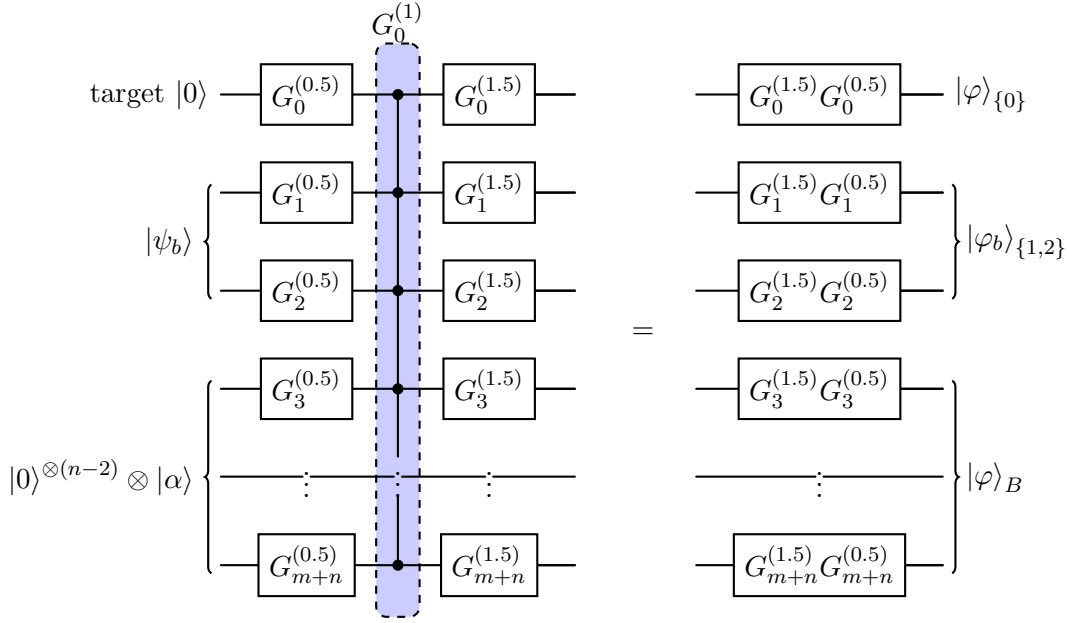
shows in a typical case how the circuit $C$ simplifies before layer~$2$ on initial state $\ket{\psi_b'}$.  Note that $\ket{\p_b}$ separates at $\{\{0\},\overline{\{0\}}\}$.

Since the initial state $\ket{\psi_b}\otimes\ket{0}^{\otimes(n-2)}$ of the input qubits has pure parity~$b$, the final state of $C$ must be of the form $\ket{b}\otimes\ket{\tau}$ for some $\ket{\tau}\in\cH_{\overline{\{0\}}}$, and thus separates at $\{\{0\},\overline{\{0\}}\}$.  It follows by running $\ket{b}\otimes\ket{\tau}$ backwards through layer $2.5$ (which contains only $1$-qubit gates) that the state $\ket{\p_b'}$ of the qubits immediately after layer~$2$ also separates at $\{\{0\},\overline{\{0\}}\}$.  Therefore, the states $\ket{\p_b}$ and $\ket{\p_b'}$ on either side of layer~$2$ both separate at $\{\{0\},\overline{\{0\}}\}$, and in particular, both states are $S$-separable.

Now applying Lemma~\ref{lem:simplify} (with $C := \{0\}$ and $D := \overline{\{0\}}$) we get that on state $\ket{\p_b}$, \ $G_0^{(2)}$ either (1) disappears, (2) simplifies to $\ctrlZ{\{0\}}$, or (3) simplifies to $\ctrlZ{\overline{\{0\}}\cap S}$.  Case~(3) is impossible because the target is not pass-through by Lemma~\ref{lem:target-cant-pass-2} and so its state is a proper superposition of $\ket{0}$ and $\ket{1}$ at layer~$2$.  Thus there are only two possibilities for $G_0^{(2)}$ given $b$: either $G_0^{(2)}$ disappears or simplifies to $\ctrlZ{\{0\}}$, which is the $1$-qubit $Z$-gate.  Therefore only two final states of the target are possible on initial state $\ket{\psi_b'}$:
\[ \ket{b} = \begin{cases}
G_0^{(2.5)}\ket{\p}_{\{0\}} & \mbox{if $G_0^{(2)}$ disappears on $\ket{\p_b}$,} \\
G_0^{(2.5)}Z\ket{\p}_{\{0\}} & \mbox{otherwise.}
\end{cases} \]
If follows that $G_0^{(2)}$ must disappear for one of $b$'s values---say $b_0$---but not the other one.  Thus we have that $G_0^{(2)}$ disappears on state $\ket{\p_{b_0}} = \ket{\p}_{\{0\}}\otimes\ket{\p_{b_0}}_{\{1,2\}}\otimes\ket{\p}_B$.  Noting that $\ket{\p_{b_0}}$ separates at $\{\{0,1,2\},B\}$, we now apply Lemma~\ref{lem:no-zero-divisors} (with $A := \{0,1,2\}$) to see that $G_0^{(2)}$ disappears --- and hence $\ket{b} = G_0^{(2.5)}\ket{\p}_{\{0\}}$ --- on $\ket{\p}_{\{0\}}\otimes\ket{\p_{b_0}}_{\{1,2\}}\otimes\ket{\sigma}_B$ for \emph{any} state $\ket{\sigma}_B\in\cH_B$.  That implies that the final state of the target does not depend on the input qubit~$3$, and so $C$ cannot weakly compute $\oplus_n$.
\end{proof}

We are now ready to prove Theorem~\ref{thm:depth-2-unclean}.  The idea of the proof is to show that $G_0^{(2)}$ must ``act classically'' on most of the input qubits.

\begin{proof}[Proof of Theorem~\ref{thm:depth-2-unclean}]
Suppose $C$ is a depth-$2$ $\QAC$-circuit that computes~$\oplus_n$ for some $n\ge 4$.  By Lemma~\ref{lem:target-cant-pass-2}, $C$'s target is not pass-through, and $G_0^{(2)} = \ctrlZ{S}$ for some $S$ that contains $0$ and at least one other qubit.  If some $\ctrlZ{}$-gate touches the target on layer~$1$, then let $T$ be such that $G_0^{(1)} = \ctrlZ{T}$; otherwise, set $T := \{0\}$.  By Lemma~\ref{lem:target-connects-only-one-input}, $T$ can include at most one input qubit.  ($T$ may contain any number of ancilla qubits, however.)  We can assume WLOG that $T\cap\{2,\ldots,n\} = \emptyset$.  For any $x\in\two^{n-1}$, define the initial state
\[ \ket{\psi_x} := \ket{0}\otimes\ket{0}\otimes\ket{x}\otimes\ket{0}^{\otimes m} \]
obtained by setting input qubit~$1$ to $\ket{0}$ and the rest of the input qubits to $\ket{x}$ (and the target and all ancilla qubits to $\ket{0}$).  Note that $\ket{\psi_x}$ is the tensor product of $1$-qubit states and hence separates at every partition of the qubits.  Let
\[ \ket{\p_x} := G^{(1.5)}G^{(1)}G^{(0.5)}\ket{\psi_x} \]
be the result of running the state $\ket{\psi_x}$ through layers~$0.5$--$1.5$ of the circuit.  It is evident that $\ket{\p_x}$ separates at $\{T,\overline{T}\}$.

\begin{claim}
Given initial state $\ket{\psi_x}$ for $x\in\two^{n-1}$, \ $G_0^{(2)}$ either disappears or simplifies to $\ctrlZ{S\cap T}$.
\end{claim}

\begin{proof}[Proof of the Claim]
If $S\subseteq T$ we are done, so assume $S\not\subseteq T$.  By assumption, running $C$ on $\ket{\psi_x}$ results in a state of the form $\ket{b_x}\otimes\ket{\tau}$, where $b_x := \oplus x$ and $\ket{\tau}\in\cH_{\overline{\{0\}}}$ is some state of the non-target qubits.  Running this state backwards through layer~$2.5$ as in the proof of Lemma~\ref{lem:target-connects-only-one-input}, we get that the state $\ket{\p_x'}$ of the qubits just after layer~$2$ separates at $\{\{0\},\overline{\{0\}}\}$ and hence is $S$-separable.  Likewise, $\ket{\p_x}$ is also $S$-separable.  By Lemma~\ref{lem:simplify}, either $G_0^{(2)}$ disappears on $\ket{\p_x}$ or simplifies to $\ctrlZ{S\cap A}$ for some subset $A$ of one of the four sets $T,\overline{T},\{0\},\overline{\{0\}}$.  Since $C$'s target is not pass-through by Lemma~\ref{lem:target-cant-pass-2}, we can assume $0\in A$, and thus $A\subseteq S\cap T$.  This implies the weaker statement that $G_0^{(2)}$ simplifies to $\ctrlZ{S\cap T}$ on $\ket{\p_x}$ in the case where $G_0^{(2)}$ does not disappear.
\end{proof}

Since $\ket{\p_x}$ separates at $\{T,\overline{T}\}$, we can write
\[ \ket{\p_x} = \ket{\p}_T\otimes\ket{\p_x}_{\overline{T}}\;, \]
where $\ket{\p}_T\in\cH_T$ does not depend on $x$ and $\ket{\p_x}_{\overline{T}}\in\cH_{\overline{T}}$.  From the Claim it follows that, given initial state $\ket{\psi_x}$, the qubits in $T$ do not entangle with any other qubits on layer~$2$ of the circuit and so can only be in one of two possible final states after layer~$2.5$:
\[ \ket{\tau_x}_T = \begin{cases}
G_T^{(2.5)}\ket{\p}_T & \mbox{if $G_0^{(2)}$ disappears on $\ket{\p_x}$,} \\
G_T^{(2.5)}(\ctrlZ{S\cap T})\ket{\p}_T & \mbox{if $G_0^{(2)}$ does not disappear on $\ket{\p_x}$,}
\end{cases} \]
unentangled with any other qubits.  Since $\ket{\tau_x}$ determines the final target value, it must change according to $\oplus x$ (because $T$ includes the target), there must exist an $x_0\in\two^{n-1}$ such that $G_0^{(2)}$ disappears on $\ket{\p_{x_0}}$.  Fix such an $x_0$.

By Lemma~\ref{lem:no-three-input-qubits}, input qubits~$2,3,4$ cannot all be touched by the same gate on layer~$1$.  Without loss of generality, we can assume that qubit~$4$ does not share a layer-$1$ gate with qubits $2$ and $3$.  This means that we can decompose $\ket{\p_{x_0}}$ further:
\[ \ket{\p_{x_0}} = \ket{\p}_T\otimes \ket{\p_{x_0}}_{T_1} \otimes \ket{\p_{x_0}}_{T_2} \]
for some partition $\{T_1,T_2\}$ of $\overline{T}$ such that $T_1$ contains qubits $2$ and $3$ and $T_2$ contains qubit $4$.  We then have that $\ket{\p_{x_0}}$ separates at $\{T\cup T_1,T_2\}$.  By Lemma~\ref{lem:no-zero-divisors}, either $G_0^{(2)}$ disappears on $\ket{\p}_T\otimes\ket{\p_{x_0}}_{T_1} \otimes \ket{\sigma}_{T_2}$ for any $\ket{\sigma}_{T_2}\in\cH_{T_2}$ or $G_0^{(2)}$ disappears on $\ket{\sigma}_{T\cup T_1} \otimes \ket{\p_{x_0}}_{T_2}$ for any $\ket{\sigma}_{T\cup T_1}\in\cH_{T\cup T_1}$.  In the former case, the final target value does not depend on input qubit~$4$; in the latter, it does not depend on input qubits~$2$ or $3$.  In either case, $C$ cannot compute $\oplus_n$.
\end{proof}

\begin{rmrk}
The condition that all the ancilla qubits are initially $\ket{0}$ in Theorem~\ref{thm:depth-2-unclean} can be relaxed to allow for a more general initial ancilla state, provided the overall initial state of the circuit separates at $\{T,T_1\cup T_2\}$ and at $\{T\cup T_1,T_2\}$.  That is, $C$ cannot $\ket{\alpha}$-compute $\oplus_n$ for any $\ket{\alpha}$ such that $\ket{0^{n+1}}\otimes\ket{\alpha}$ separates at $\{T,T_1\cup T_2\}$ and at $\{T\cup T_1,T_2\}$.
\end{rmrk}

\begin{rmrk}
Theorem~\ref{thm:depth-2-unclean} also holds for depth-$2$ circuits that include $G_\eta$-gates as in Eq.~(\ref{eqn:G-eta}), and the value of $\eta$ need not be the same for each gate.
\end{rmrk}

\subsection{Further Research}

Our techniques currently work for depth~$2$, but obviously, we would like to prove limitations on $\QAC$-circuits of higher depth.
We hope the entanglement lemma (Lemma~\ref{lem:S-entangled}) will be useful for depth~$3$ and beyond, however.  Lemma~\ref{lem:kill-parity} is stronger than needed for the current results; by committing clusters of input qubits to certain states, we can turn off $\ctrlZ{}$-gates through more than two layers.  These two lemmas as well as Lemma~\ref{lem:no-zero-divisors} provide powerful tools for dealing with $\QAC$-circuits of higher depth.  By simplifying a circuit in the right way, one hopes to reduce its effective depth, and this in turn may lead to an inductive proof of the limitations of such circuits.

More specifically, Lemma~\ref{lem:S-entangled} may be useful for depth~$3$ and beyond because it disallows many different circuit topologies for $\QAC$-circuits computing parity.  For example, the following circuit topology is impossible for simulating parity (or any classical reversible function for that matter) cleanly unless the middle gate simplifies:
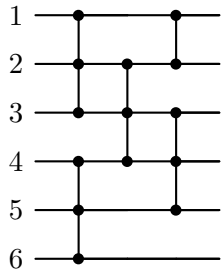
\begin{center}
\begin{quantikz}
\lstick{$1$} & \ctrl{2}   & \qw        & \ctrl{1}   & \qw \\
\lstick{$2$} & \control{} & \ctrl{2}   & \control{} & \qw \\
\lstick{$3$} & \control{} & \control{} & \ctrl{2}   & \qw \\
\lstick{$4$} & \ctrl{2}   & \control{} & \control{} & \qw \\
\lstick{$5$} & \control{} & \qw        & \control{} & \qw \\
\lstick{$6$} & \control{} & \qw        & \qw        & \qw
\end{quantikz}
\end{center}
(Here only the $\ctrlZ{}$-gates are shown; the single qubit gates are suppressed.)
The reason is that, for any classical input, the state on the far left is completely separable, and so the state immediately after the first layer is $\{2,3,4\}$-separable (separating at $\{\{1,2,3\},\{4,5,6\}\}$).  If the middle gate does not simplify, then by the lemma, the state $\ket{\psi}$ immediately to its right must be $\{2,3,4\}$-entangled.  Now assuming a clean simulation, the state on the far right is completely separable, and so running the circuit backwards from the right, we see that $\ket{\psi}$ must be $\{2,3,4\}$-separable (separating at $\{\{1,2\},\{3,4,5,6\}\}$).  This is a contradiction.

Finally, we note that the techniques used to prove that parity cannot be computed by classical $\AC^0$-circuits (i.e., random restrictions and switching lemmas) are not necessarily needed or even relevant here, because fanout is taken for granted in the classical case, unlike in the quantum case.

\section*{Acknowledgments}

The authors would like to thank Alexander Duncan for helpful discussions regarding the results in Section~\ref{sec:irred-results}.

\bibliographystyle{alpha}

\begin{thebibliography}{GHMP02}

\bibitem[ADOY24]{ADOY:QAC0-superlinear-ancillae}
Anurag Anshu, Yangjing Dong, Fengning Ou, and Penghui Yao.
\newblock On the computational power of {QAC0} with barely superlinear
  ancillae, 2024.

\bibitem[Ajt83]{Ajtai:AC0}
M.~Ajtai.
\newblock {$\Sigma^1_1$} formul\ae\ on finite structures.
\newblock {\em Annals of Pure and Applied Logic}, 24:1--48, 1983.

\bibitem[BGK18]{BGK:quantum-advantage}
S.~Bravyi, D.~Gosset, and R.~K\"{o}nig.
\newblock Quantum advantage with shallow circuits.
\newblock {\em Science}, 362(6412):308--311, 2018.

\bibitem[FFG{\etalchar{+}}06]{FFGHZ:fanout}
M.~Fang, S.~Fenner, F.~Green, S.~Homer, and Y.~Zhang.
\newblock Quantum lower bounds for fanout.
\newblock {\em Quantum Information and Computation}, 6:46--57, 2006.

\bibitem[FGHZ05]{FGHZ:constant-depth}
S.~Fenner, F.~Green, S.~Homer, and Y.~Zhang.
\newblock Bounds on the power of constant-depth quantum circuits.
\newblock In {\em Proceedings of the 15th International Symposium on
  Fundamentals of Computation Theory}, volume 3623 of {\em Lecture Notes in
  Computer Science}, pages 44--55. Springer-Verlag, 2005.

\bibitem[FSS84]{FSS:AC0}
M.~Furst, J.~B. Saxe, and M.~Sipser.
\newblock Parity, circuits, and the polynomial time hierarchy.
\newblock {\em Mathematical Systems Theory}, 17:13--27, 1984.

\bibitem[GHMP02]{GHMP:QACC}
F.~Green, S.~Homer, C.~Moore, and C.~Pollett.
\newblock Counting, fanout and the complexity of quantum {ACC}.
\newblock {\em Quantum Information and Computation}, 2:35--65, 2002.

\bibitem[GM24]{GM:QAC-threshold}
Daniel Grier and Jackson Morris.
\newblock Quantum threshold is powerful, 2024.

\bibitem[H{\v{S}}05]{HS:fanout}
P.~H{\o}yer and R.~{\v{S}}palek.
\newblock Quantum fan-out is powerful.
\newblock {\em Theory of Computing}, 1(5):81--103, 2005.

\bibitem[KLM07]{KLM:quantum-book}
P.~Kaye, R.~Laflamme, and M.~Mosca.
\newblock {\em An Introduction to Quantum Computing}.
\newblock Oxford University Press, 2007.

\bibitem[KSV02]{KSV:quantum-book}
A.~Yu.\ Kitaev, A.~H. Shen, and M.~N. Vyalyi.
\newblock {\em Classical and quantum computation}.
\newblock American Mathematical Society, Providence, RI, 2002.

\bibitem[Moo99]{Moore:fanout}
C.~Moore.
\newblock Quantum circuits: Fanout, parity, and counting, 1999.
\newblock arXiv:quant-ph/9903046.

\bibitem[NC00]{NC:quantumbook}
M.~A. Nielsen and I.~L. Chuang.
\newblock {\em Quantum Computation and Quantum Information}.
\newblock Cambridge University Press, 2000.

\bibitem[NPVY24]{NPVY:QAC0-Pauli-Spectrum}
Shivam Nadimpalli, Natalie Parham, Francisca Vasconcelos, and Henry Yuen.
\newblock On the pauli spectrum of {QAC0}.
\newblock In {\em Proceedings of the 56th Annual ACM Symposium on Theory of
  Computing}, pages 1498--1506, New York, NY, USA, 2024. Association for
  Computing Machinery.

\bibitem[PFGT20]{PFGT:fanout}
D.~Pad\'{e}, S.~Fenner, D.~Grier, and T.~Thierauf.
\newblock Depth-2 {QAC} circuits cannot simulate quantum parity, 2020.
\newblock arXiv:2005.12169.

\bibitem[Piu14]{Pius:QAC}
Einar Pius.
\newblock {\em Parallel Quantum Computing From Theory to Practice}.
\newblock PhD thesis, The University of Edinburgh, 8 2014.

\bibitem[Ros21]{Rosenthal:parity}
G.~Rosenthal.
\newblock Bounds on the {$\textup{QAC}^0$} complexity of approximating parity.
\newblock In James~R. Lee, editor, {\em 12th Innovations in Theoretical
  Computer Science Conference (ITCS)}, number~32 in Leibniz International
  Proceedings in Informatics (LIPIcs), pages 32:1--32:20, 2021.
\newblock arXiv:2008.07470.

\bibitem[SV10]{SV:indecomposable}
Amir Shpilka and Ilya Volkovich.
\newblock On the relation between polynomial identity testing and finding
  variable disjoint factors.
\newblock In Samson Abramsky, Cyril Gavoille, Claude Kirchner, Friedhelm Meyer
  auf~der Heide, and Paul~G. Spirakis, editors, {\em Automata, Languages and
  Programming}, pages 408--419, Berlin, Heidelberg, 2010. Springer Berlin
  Heidelberg.

\bibitem[TT16]{TT:constant-depth-collapse}
Y.~Takahashi and S.~Tani.
\newblock Collapse of the hierarchy of constant-depth exact quantum circuits.
\newblock {\em Computational Complexity}, 25(4):849--881, 2016.
\newblock Conference version in Proceedings of the 28th IEEE Conference on
  Computational Complexity (CCC 2013).

\end{thebibliography}
\newcommand{\etalchar}[1]{$^{#1}$}

\end{document}